\documentclass[pre,twocolumn,10pt,aps,longbibliography]{revtex4-1}

 \usepackage[utf8]{inputenc}
 \usepackage[greek,english]{babel}
 \usepackage{bm}
 \usepackage{verbatim}
 \usepackage{amsmath}
 \usepackage{amssymb}
 \usepackage{amsthm}
 \usepackage{latexsym}
 \usepackage{amsfonts}
 \usepackage{epsfig}
 \usepackage{epstopdf}
 \usepackage{wasysym} 
 \usepackage{subfigure}
 \usepackage{color}
 \definecolor{darkblue}{rgb}{0,0,.5}
 \definecolor{BLUE}{rgb}{0,0,1}
 \definecolor{BLACK}{rgb}{0,0,0}
 \usepackage[linktocpage, colorlinks=true, linkcolor=darkblue, citecolor=darkblue]{hyperref}
 \usepackage[all]{hypcap}
 \usepackage[makeroom]{cancel}

\newcommand{\C}[1]{{\cal{#1}}}
\newcommand{\bb}[1]{\textbf{#1}}
\newcommand{\lr}[1]{{\left\langle {#1}\right\rangle}}
\newcommand{\rl}[0]{{\rangle\langle}}

\def\dbar{{\mathchar'26\mkern-12mu d}}

\begin{document}

\title{First and Second Law of Quantum Thermodynamics: \\ 
A Consistent Derivation Based on a Microscopic Definition of Entropy}

\author{Philipp Strasberg$^1$}
\author{Andreas Winter$^{1,2}$}
\affiliation{$^1$F\'isica Te\`orica: Informaci\'o i Fen\`omens Qu\`antics, Departament de F\'isica, Universitat Aut\`onoma de Barcelona, 08193 Bellaterra (Barcelona), Spain}
\affiliation{$^2$ICREA -- Instituci\'o Catalana de Recerca i Estudis Avan\c{c}ats, Pg.~Lluis Companys, 23, 08010 Barcelona, Spain}

\date{\today}

\begin{abstract}
 Deriving the laws of thermodynamics from a microscopic picture is a central quest of statistical mechanics. This 
 tutorial focuses on the derivation of the first and second law for isolated and open quantum systems far from 
 equilibrium, where such foundational questions also become practically relevant for emergent nanotechnologies. 
 The derivation is based on a microscopic definition of five essential quantities: internal energy, thermodynamic 
 entropy, work, heat and temperature. These definitions are shown to satisfy the phenomenological laws of 
 nonequilibrium thermodynamics for a large class of states and processes. The consistency with previous results is 
 demonstrated. The framework applies to multiple baths including particle transport and accounts for processes with, 
 e.g., a changing temperature of the bath, which is determined microscopically. An integral fluctuation theorem 
 for entropy production is satisfied. In summary, this tutorial introduces a consistent and versatile framework 
 to understand and apply the laws of thermodynamics in the quantum regime and beyond. 
\end{abstract}

\maketitle

\newtheorem{lemma}{Lemma}[section]

\section{Introduction to Nonequilibrium Thermodynamics: Phenomenology}
\label{sec noneq thermo}

Before we turn to any microscopic derivation of the laws of thermodynamics, it seems worthwhile to briefly recall 
\emph{what} we actually want to derive.

Thermodynamics is an independent physical theory, whose principles have been applied with an 
enormous success over a wide range of length, time and energy scales.  
It arose out of the desire to understand transformations of matter in chemistry and 
engineering in the 19th century~\footnote{Sometimes it is asserted that thermodynamics played an important role for 
the \emph{industrial revolution} to design efficient heat engines. Historically speaking, this is incorrect. The 
industrial revolution is associated with the period from 1760 to (at most) 1840 (the steam engine of Watt was 
introduced in 1776). The first modern work on thermodynamics is perhaps due to Carnot in 1824, who, however, was not 
read by his contemporaries. The first law of thermodynamics was established around 1850 and the modern formulation of 
the second law goes back to Clausius in 1865. Even then, however, it took time until engineers were inspired by 
theoretical insights from thermodynamics. To the best of our knowledge, Diesel (at the end of the 19th century) 
patented the first engine which was based on the insight that a high temperature gradient increases the efficiency of 
the engine.}. 
The systems under investigation were macroscopic and described by very few variables; for instance, 
temperature $T$, pressure $p$ and volume $V$. These macroscopic systems could exchange heat $Q$ with their surroundings 
and mechanical work $W$ could be supplied to them. 
A prototypical example of a thermodynamic setup partitioned into a system, a heat bath and a work reservoir is 
sketched in Fig.~\ref{fig setup thermo}. 

The theory is based on two central axioms, which are called the first and second law of 
thermodynamics~\cite{FlandersSwann:1st-and-2nd, KondepudiPrigogineBook2007} (there is also a zeroth and a third law of 
thermodynamics, which are not the topic of this paper). The first law states that the change $\Delta U_S$ in internal 
energy of the system is balanced by heat $Q$ and mechanical work $W$: 
\begin{equation}\label{eq 1st law}
 \Delta U_S = Q + W.
\end{equation}
Note that we define heat and work to be positive if they increase the internal energy of the system. The first 
law is a consequence of conservation of energy applied to the system, the heat bath and the work reservoir. However, 
the fundamental distinction between heat and work becomes only transparent by considering the second law. 

\begin{figure}[b]
 \centering\includegraphics[width=0.30\textwidth,clip=true]{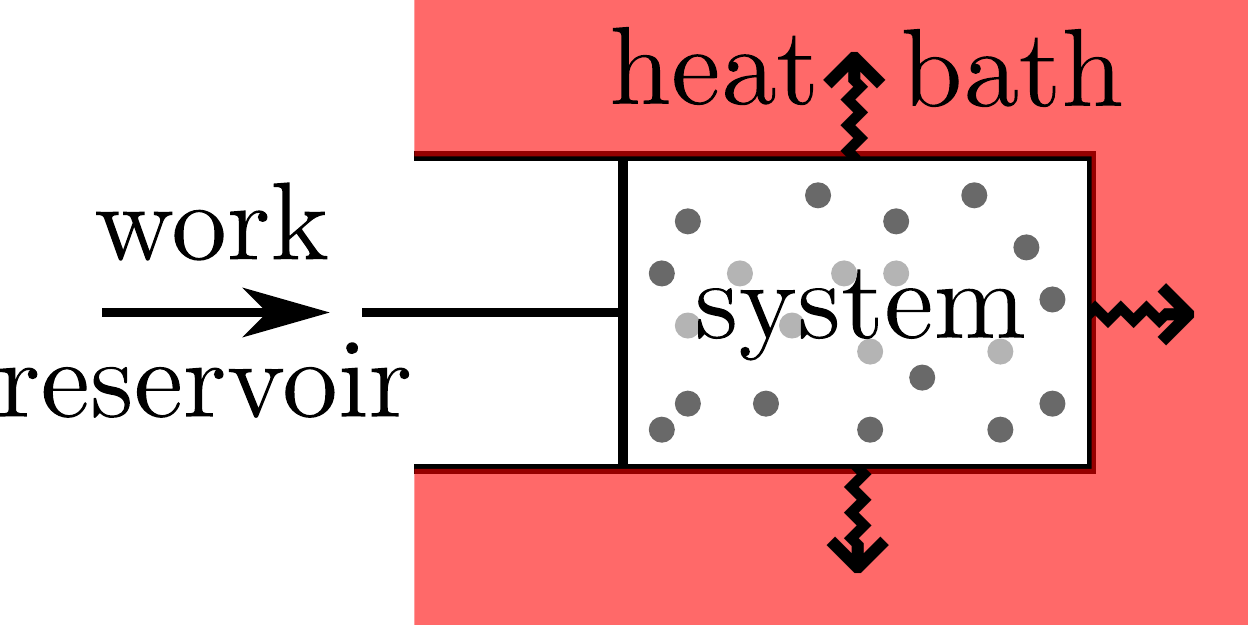}
 \label{fig setup thermo} 
 \caption{Thermodynamic setup where the system is a gas in a container. By pushing a piston, the 
 thermodynamic variables (such as $T$, $p$ or $V$) can be changed in a mechanically controlled way, which is abstracted 
 as the action of a `work reservoir.' Furthermore, through the walls of the container the gas is in simultaneous 
 contact with a heat bath, with which it can exchange energy. This exchange of energy is accompanied with an exchange 
 of entropy, which is the defining property to call this energy exchange `heat.' }
\end{figure}

The second law, in its most general form, states that ``the entropy of the universe tends to a 
maximum''~\cite{Clausius1865}. In equations, for any physical process 
\begin{equation}\label{eq EP basic basic}
 \Delta S_\text{univ} \ge 0, \tag{2a}
\end{equation}
where $S_\text{univ}$ denotes the \emph{thermodynamic} entropy of the universe, which should be distinguished from any 
information theoretic notion of entropy at this point. Note that the terminology `universe' does not necessarily refer 
to the entire universe in the cosmological sense, but to any system which is sufficiently isolated from the rest of 
the world. For our purposes, this also includes a gas of ultracold atoms~\cite{BlochDalibardZwergerRMP2008, 
LewensteinSanperaAhufingerBook2012} or the system \emph{and} the bath within the open quantum systems 
paradigm~\cite{BreuerPetruccioneBook2002, DeVegaAlonsoRMP2017}. The change in entropy of the universe is often called 
the \emph{entropy production}~\cite{KondepudiPrigogineBook2007} and denoted by $\Sigma = \Delta S_\text{univ}$. 
If $\Sigma = 0$, the process is called \emph{reversible}, otherwise \emph{irreversible}. 

Focusing on the system-bath setup, e.g., as sketched in Fig.~\ref{fig setup thermo}, the entropy of the universe 
is often additively decomposed into the entropy of the system and the environment: 
$S_\text{univ} = S_S + S_\text{env}$. This is justified whenever surface effects are negligible compared to bulk 
properties, which is often (but not always) the case for macroscopic systems. Then, the second law becomes 
\begin{equation}\label{eq EP basic}
 \Sigma = \Delta S_S + \Delta S_\text{env} \ge 0. \tag{2b}
\end{equation}
Furthermore, the environment is typically assumed to be well-described by an equilibrium state with a 
time-dependent temperature $T$ such that its change in entropy is $\Delta S_\text{env} = -\int\dbar Q/T$. Here, 
$\dbar Q$ denotes an infinitesimal heat flow into the system. Then, the second law reads 
\begin{equation}\label{eq EP traditional}
 \Sigma = \Delta S_S - \int\frac{\dbar Q}{T} \ge 0, \tag{2c}
\end{equation}
which was introduced by Clausius, who called $\Sigma$ \emph{uncompensated transformations} (``unkompensierte 
Verwandlungen'')~\cite{Clausius1865}. In fact, the word `entropy' was chosen by Clausius based on the ancient greek 
word for `transformation' (\textgreek{τροπή}). Equation~(\ref{eq EP traditional}) is often referred to as Clausius' 
inequality. Finally, if the bath gets only slightly perturbed away from its initial temperature, here denoted by 
$T_0$, then Eq.~(\ref{eq EP traditional}) reduces to 
\begin{equation}\label{eq EP traditional OQS}
 \Sigma = \Delta S_S - \frac{Q}{T_0} \ge 0 \tag{2d}
\end{equation}
with $Q = \int\dbar Q$ the total flow of heat from the bath. 

These basic building blocks of phenomenological nonequilibrium thermodynamics can be further extended to, e.g., 
multiple heat baths or particle transport (above, we tacitly assumed that the system only exchanges energy but not 
particles with the bath). For most parts, we focus on the microscopic derivations of the laws above and turn to these 
extensions only at the end. 

\section{Goal of this tutorial}
\label{sec goal}
\setcounter{equation}{2}

\subsection{The need for a microscopic derivation}

While it is important to emphasize the status of thermodynamics as an independent physical theory, its precise scope is 
debated and problems appear when trying to apply it far from equilibrium. It thus remains a subject of ongoing 
research~\cite{LebonJouCasasVazquezBook2008}. 

The difficulties one faces with a purely phenomenological approach are perhaps best exemplified by the notion of system 
entropy $S_S$. How should this quantity---apart from an unimportant additive constant---be defined out of equilibrium? 
Clausius suggested to use Eq.~(\ref{eq EP traditional}) by postulating that any two system states can be connected by 
a \emph{reversible} transformation~\cite{Clausius1865}. If such a transformation is found, 
inequality~(\ref{eq EP traditional}) becomes an equality, 
\begin{equation}\label{eq Clausius reversible}
 \Delta S_S = \int_R \frac{\dbar Q}{T},
\end{equation}
where the subscript $R$ means `reversible.' Equation~(\ref{eq Clausius reversible}) allows to quantify $\Delta S_S$ by 
measuring the time-dependent temperature $T$ and by computing $\dbar Q = \C C_B(T)dT$, where $\C C_B(T)$ is the heat 
capacity of the bath. Unfortunately, on phenomenological grounds it is not known how to 
construct such reversible transformations connecting nonequilibrium states in general, and it seems doubtful that this 
is always possible. Widely accepted solutions to this problem seem to exist only in the linear response 
regime~\cite{PottierBook2010} or if the local equilibrium assumption is valid~\cite{KondepudiPrigogineBook2007}. 

In this tutorial, we mostly focus on small systems, which can show quantum effects, are driven far from 
equilibrium, and are in contact with an environment. Such systems are called open quantum 
systems~\cite{BreuerPetruccioneBook2002, DeVegaAlonsoRMP2017}. For many potential future technologies---such as 
thermoelectric devices, solar cells, energy efficient computers, refrigerators that cool down to almost zero Kelvin, 
or quantum computing, sensing or communication devices---these are very interesting systems. Furthermore, we are also 
interested in isolated quantum many-body systems such as ultracold quantum gases~\cite{BlochDalibardZwergerRMP2008, 
LewensteinSanperaAhufingerBook2012}. In all of these cases, neither the local equilibrium assumption nor linear 
response theory can be applied in general. A thermodynamic description purely based on phenomenological grounds 
therefore appears challenging. 

Moreover, the traditionally used classifications in thermodynamics of a heat bath and a work reservoir are becoming 
increasingly inadequate. Nowadays, experimentalists have access to information beyond simple macroscopic parameters 
such as temperatures or chemical potentials, they can engineer specifically tailored environments and make use of 
more sophisticated external resources, including quantum measurements and feedback control loops. Accounting for all 
these possibilities in a purely phenomenological way seems impossible. 

Finally, a microscopic derivation of the laws of thermodynamics gives us a better understanding about the 
phenomenological theory. The resulting theoretical framework, in which thermodynamic principles are explained and 
supplemented by quantum mechanical and statistical considerations, is called \emph{quantum thermodynamics}. 

\subsection{Setting}

We briefly recall the quantum mechanical setting we are interested in. First, we consider the case of 
an isolated system. Its state at time $t$ is described by a density matrix $\rho(t)$ and the Hamiltonian of the system 
is denoted $H(\lambda_t)$. Here, $\lambda_t$ is some externally specified driving protocol (e.g., a changing 
electromagnetic field or the moving piston in Fig.~\ref{fig setup thermo}). The validity of modelling the dynamics of 
a quantum system via a time-dependent Hamiltonian rests on the assumption that the driving field is generated by 
a classical, macroscopic device. 
Finally, the dynamics of the system state obeys the Liouville-von Neumann equation ($\hbar\equiv 1$ throughout) 
\begin{equation}\label{eq Liouville von Neumann}
 \frac{\partial}{\partial t}\rho(t) = -i[H(\lambda_t),\rho(t)],
\end{equation}
where $[A,B] = AB-BA$ is the commutator. The time evolution starting from an initial state 
$\rho(0)$ (we always set the initial time to be $t=0$) is therefore unitary: 
\begin{equation}\label{eq time evo isolated}
 \rho(t) = U(t,0)\rho(0)U^\dagger(t,0).
\end{equation}
Here, the unitary time evolution operator $U(t,0) = \exp_+[-i\int_0^t ds H(\lambda_s)]$ is defined as the time-ordered 
exponential of the Hamiltonian. Notice that we make \emph{no} assumption about the specific form of $H(\lambda_t)$ in 
the following, we only need to make one assumption about the initial state $\rho(0)$, as explained in the next section. 

Next, we consider open quantum systems and use a subscript $SB$ (for system and bath) to denote the global 
state and Hamiltonian. The latter is of the form 
\begin{equation}
 \begin{split}\label{eq system bath Hamiltonian}
  H_{SB}(\lambda_t) &= H_S(\lambda_t)\otimes 1_B + 1_S\otimes H_B + V_{SB} \\
                    &= H_S(\lambda_t) + H_B + V_{SB},
 \end{split}
\end{equation}
where we suppressed tensor products with the identity in the notation of the second line. Here, $H_S$ ($H_B$) denotes 
the Hamiltonian of the unperturbed system (bath) and $V_{SB}$ their interaction. Again, they are completely 
arbitrary in our framework. However, in view of Fig.~\ref{fig setup thermo}, we assumed that the external driving 
protocol $\lambda_t$ only influences the system Hamiltonian. It is also possible to consider time-dependent interactions 
$V_{SB}(\lambda_t)$ to model, e.g., the coupling and decoupling of the system and the bath. Our results 
continue to hold in this case, but for ease of presentation we assume $V_{SB}$ to be time-independent. Finally, while 
the joint system-bath state $\rho_{SB}(t)$ evolves in time according to Eq.~(\ref{eq Liouville von Neumann}) with 
respect to the Hamiltonian~(\ref{eq system bath Hamiltonian}), the evolution of the reduced system state 
\begin{equation}
 \rho_S(t) = \mbox{tr}_B\{\rho_{SB}(t)\}
\end{equation}
(with $\mbox{tr}_B\{\dots\}$ denoting the partial trace over the bath degrees of freedom) 
is no longer unitary. In fact, the evolution of $\rho_S(t)$ is markedly different and in general very hard to 
compute~\cite{BreuerPetruccioneBook2002, DeVegaAlonsoRMP2017}. The laws of thermodynamics derived below hold, however, 
regardless of these considerations. 

A final word on terminology is useful to avoid confusion. In thermodynamics, a system is called (i)~isolated, 
(ii)~closed or (iii)~open if it can exchange (i)~only work, (ii)~work and heat in form of energy or (iii)~work and 
heat in form of energy and particles with its surroundings. In contrast, in open quantum system theory the words 
isolated and closed are used interchangeably to describe case (i), whereas cases (ii) and (iii) are both called 
open. We indeed use the terminology open in the latter sense and, for definiteness, call case (i) isolated. 

\subsection{Desiderata and assumption}

We here precisely specify what we mean by a consistent microscopic derivation of the laws of thermodynamics 
and what we need to assume to accomplish it. 

First, in Sec.~\ref{sec noneq thermo} we saw that there are five important concepts in phenomenological thermodynamics. 
These are the two state functions internal energy and thermodynamic entropy, the two process-dependent quantities 
mechanical work and heat and the temperature appearing in Clausius' inequality~(\ref{eq EP traditional}). 
For all of them we like to provide a microscopic definition, which is expressed in terms of the quantum mechanical 
Hamiltonian and the density matrix (or quantities derived from them). 

Second, these quantities are supposed to satisfy the first law~(\ref{eq 1st law}) as well as the second 
laws~(\ref{eq EP basic basic}),~(\ref{eq EP basic}),~(\ref{eq EP traditional}) and~(\ref{eq EP traditional OQS}). 
As explained in the previous section, Eq.~(\ref{eq EP basic basic}) is more general than Eq.~(\ref{eq EP basic}), 
which is more general than Eq.~(\ref{eq EP traditional}), which is more general than Eq.~(\ref{eq EP traditional OQS}), 
with each one following from the previous one in its respective range of validity, and we demand that this 
\emph{hierarchy} of second laws is reproduced in the microscopic derivation. We remark that, due to the relations 
extablished by the laws of thermodynamics, the five thermodynamic quantities we seek to define are not all independent. 

Third, as an important consistency check, we demand that the proposed definitions should reduce to well known results 
derived previously in and out of equilibrium. 

The above three criteria are certainly the most basic desiderata we can have about any microscopic derivation of the 
laws of thermodynamics. As it turns out, it is possible to strictly satisfy all of them for any Hamiltonian 
of the isolated system or the system-bath composite. 

We need, however, one assumption about the initial state. This assumption is mathematically specified 
later on, but here we explain \emph{why} we need one. The microscopic equations of motion, such as 
Eq.~(\ref{eq Liouville von Neumann}) or Newton's equation for classical systems, obey a property called 
\emph{time-reversal symmetry}. Roughly speaking, this means that to any solution of the dynamics with a given initial 
and final condition, it is possible to find a conjugate `twin solution' with initial and final condition 
\emph{exchanged} (Appendix~\ref{sec time reversal} gives a precise account of time-reversal symmetry). 
Thus, if thermodynamic entropy increases for the first solution, it must decrease for the conjugate 
twin solution. Consequently, ``the second law can never be proved mathematically by means of the equations of dynamics 
alone,'' as Boltzmann stressed already~\cite{BoltzmannNature1895}. 

The reason why we see no violations of the second law in our daily life comes from the fact that initial conditions, 
which generate a spontaneous entropy decrease, are extremely hard to prepare experimentally, see 
Fig.~\ref{fig entropy decrease} for an illustration. Mathematically, these `unnatural' states, which are very 
hard to prepare, need to be excluded by a proper choice of initial state specified later on. 

\begin{figure}
 \centering\includegraphics[width=0.49\textwidth,clip=true]{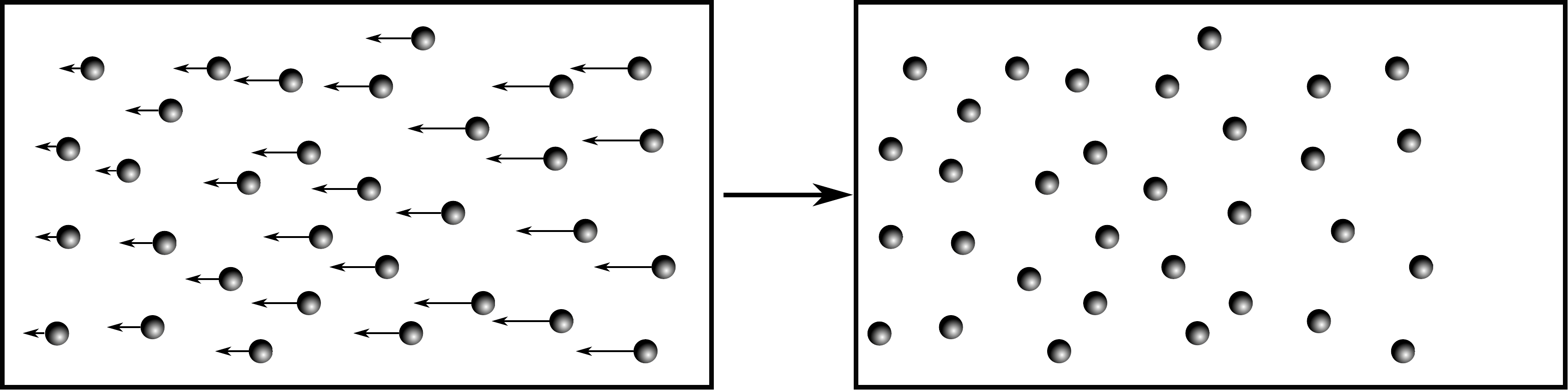}
 \label{fig entropy decrease} 
 \caption{Time evolution of gas particles in a box with perfectly reflecting walls. \emph{Left:} Initially, all gas 
 particles have a velocity pointing to the left such that in the next time step none of them is reflected to the right. 
 This is an extremely unlikely state and an experimental preparation of it requires precise control about every single 
 gas particle. \emph{Right:} Given the initial state on the left, the state of the gas after the time step is 
 characterized by a lower entropy, in \emph{seeming} violation to the second law of thermodynamics. }
\end{figure}

We remark that this is not the only way to derive the second law microscopically. It is also possible to consider 
\emph{arbitrary} initial states, but in this case the second law can only be established by imposing 
restrictions on the Hamiltonian. We do not consider this approach in detail here, but see 
Refs.~\cite{VonNeumann1929, VonNeumannEPJH2010, GoldsteinEtAlEPJH2010, GoldsteinHaraTasakiArXiv2013, JinEtAlPRE2016, 
IyodaKanekoSagawaPRL2017, GemmerKnipschildSteinigewegArXiv2017} for an exposition of historical and recent 
ideas and discussions in this direction. We remark that the five thermodynamic quantities defined below nonetheless 
remain meaningful for this different approach. 

Finally, one might wonder how the second law can emerge at all in a universe with time-reversal symmetric evolution 
equations. The most likely explanation is that the universe started off in a state with extremely low entropy. Thus, 
the second law seems to be a consequence of the boundary conditions. This conjecture is known as the \emph{past 
hypothesis}. An informal discussion of the microscopic origin of the second law and the arrow of time is given in 
Ref.~\cite{LebowitzPhysToday1993}. 

\subsection{Outline}

We start with the definition of internal energy and mechanical work in isolated systems in 
Sec.~\ref{sec energy work}, which are clearly the most uncontroversial definitions. 
Afterwards, we review various microscopic notions of thermodynamic entropy for an isolated system in 
Sec.~\ref{sec entropies} and we argue for a concept called \emph{observational entropy} as the most appropriate 
candidate. 
Equipped with this concept, we then establish the second law of thermodynamics for isolated systems in 
Sec.~\ref{sec entropy second law isolated}. Also the notion of an effective nonequilibrium temperature is introduced 
there. 
This is followed by a detailed derivation of the laws of thermodynamics in open systems in 
Sec.~\ref{sec first second law open}. 
In Sec.~\ref{sec further extensions} we report on further extensions of our framework, including the treatment of 
multiple heat baths and particle exchanges. 
Section~\ref{sec fluctuation theorems} is devoted to the derivation of fluctuation theorems, which generalize previous 
results by extending the notion of entropy production to single stochastic trajectories recorded in an experiment. 
The final Sec.~\ref{sec final} contains some concluding reflections. 
Two appendices about time-reversal symmetry (Appendix~\ref{sec time reversal}) and basic information theory concepts 
(Appendix~\ref{sec basics information theory})  accompany this tutorial for self-containedness. 

\section{Internal energy and mechanical work in isolated systems}
\label{sec energy work}

For an isolated system we identify the expectation value of its Hamiltonian with the internal energy appearing in 
phenomenological thermodynamics, 
\begin{equation}\label{eq def internal energy}
 U(t) \equiv \mbox{tr}\{H(\lambda_t)\rho(t)\}.
\end{equation}
We remark that definition~(\ref{eq def internal energy}) is an assumption, but we are not aware of any attempt to 
define internal energy differently. 

If the system is not driven ($\dot\lambda_t = 0$), its internal energy is conserved since the Hamiltonian is a 
constant of motion: $\Delta U(t) = 0$. Here and in the following, we use the notation $\Delta X(t) = X(t) - X(0)$ to 
denote the change of any time-dependent state function $X(t)$. 
If the system is driven, its internal energy can change in time: 
\begin{equation}\label{eq internal energy change}
 \Delta U(t) = \mbox{tr}\{H(\lambda_t)\rho(t)\} - \mbox{tr}\{H(\lambda_0)\rho(0)\}.
\end{equation}
Since the system is isolated (i.e., only coupled to a work reservoir), no heat is flowing ($Q=0$) and 
the phenomenological first law~(\ref{eq 1st law}) forces us to identify the change in internal energy with 
the work supplied to the system: 
\begin{equation}\label{eq def work}
 \Delta U(t) = W(t).
\end{equation}
This is the first law of thermodynamics for an isolated system. A quick calculation, using 
Eq.~(\ref{eq Liouville von Neumann}) and that the trace is cyclic, reveals that the work can be expressed as 
\begin{equation}
 \begin{split}\label{eq def work and power}
  W(t) &= \int_0^t ds \frac{d}{ds} \mbox{tr}\{H(\lambda_s)\rho(s)\} \\
       &= \int_0^t ds \mbox{tr}\left\{\frac{\partial H(\lambda_s)}{\partial s}\rho(s)\right\} = \int_0^t ds \dot W(s)
 \end{split}
\end{equation}
with the instantaneously supplied power $\dot W(s)$. 

To conclude, the identifcation of mechanical work in an isolated system follows solely from the phenomenological 
first law together with the assumption~(\ref{eq def internal energy}). 

\section{Microscopic notions of thermodynamic entropy}
\label{sec entropies}

The central concept in both, thermodynamics and statistical mechanics alike, is \emph{entropy}. Unfortunately, it is 
also the most debated concept, which got constantly mystified during the history of science. Here, we review three 
important entropy concepts in statistical mechanics: the Gibbs-Shannon-von Neumann entropy, the Boltzmann entropy and 
observational entropy. The last candidate unifies the previous two concepts and it will be our microscopic choice for 
thermodynamic entropy in the following, in and out of equilibrium. We argue that this choice resonates with recent 
findings in nonequilibrium statistical mechanics and extends ideas expressed by Boltzmann, Gibbs, von Neumann, 
Wigner, Jaynes, among others. 

\subsection{Gibbs-Shannon-von Neumann entropy}

The Gibbs-Shannon-von Neumann entropy of a state $\rho$, regardless whether it is in or out of equilibrium, reads 
\begin{equation}\label{eq def von Neumann entropy}
 S_\text{vN}(\rho) \equiv -\mbox{tr}\{\rho\ln\rho\}.
\end{equation}
Since we are interested in quantum systems throughout this manuscript, we used a subscript `vN' and mostly call 
it von Neumann entropy for brevity. 

The success of Eq.~(\ref{eq def von Neumann entropy}) for applications in (classical and quantum) information and 
communication theory is undeniable~\cite{CoverThomasBook1991, NielsenChuangBook2000, ManciniWinterBook2020}. 
It has many useful properties and many information theory concepts are directly related to 
it; those which are useful for the purposes of the present manuscript are reviewed in 
Appendix~\ref{sec basics information theory}. 

One of these properties says that the von Neumann entropy is invariant under unitary evolution, i.e., for any state 
$\rho$ and any unitary $U$ 
\begin{equation}\label{eq von Neumann entropy conservation}
  S_\text{vN}(\rho) = S_\text{vN}(U\rho U^\dagger).
\end{equation}
Consequently, if we were to interpret von Neumann entropy as thermodynamic entropy (times the Boltzmann factor $k_B$, 
which we set to one in the following), then, from Eq.~(\ref{eq time evo isolated}), we would conclude that the 
thermodynamic entropy of every isolated system is always constant. This conflicts with empirical facts, which show that 
most spontaneous processes are accompanied with a \emph{strict increase} in thermodynamic entropy (e.g., the free 
expansion of a gas, the mixing of liquids or the evolution of the cosmological universe). 

Thus, von Neumann entropy can not correspond to thermodynamic entropy \emph{in general}. This fact was clearly 
recognized by von Neumann himself, who confessed that Eq.~(\ref{eq def von Neumann entropy}) is ``not applicable'' to 
problems in statistical mechanics as it is ``computed from the perspective of an observer who 
can carry out all measurements that are possible in principle''~\cite{VonNeumann1929} 
(translated in Ref.~\cite{VonNeumannEPJH2010}). 

Importantly, we did \emph{not} say that von Neumann entropy \emph{never} coincides with thermodynamic entropy. In fact, 
it does so in two important cases. 

The first case corresponds to a system at equilibrium, which can be described by the Gibbs ensemble 
\begin{equation}\label{eq canonical}
 \pi(\beta) \equiv \frac{e^{-\beta H}}{\C Z(\beta)}, ~~~ \C Z(\beta) \equiv \mbox{tr}\{e^{-\beta H}\},
\end{equation}
or a generalization thereof, e.g., the grand canonical ensemble if particle numbers are important. In this case, 
$S_\text{vN}[\pi(\beta)]$ coincides with thermodynamic entropy. We remark that this conclusion is only valid if the 
system obeys the \emph{equivalence of ensembles}~\cite{TouchetteJSP2015}. Beyond that, even the foundations of 
equilibrium statistical mechanics remain debated (see, e.g., Refs.~\cite{DunkelHilbertNP2014, CampisiPRE2015, 
AbrahamPrenrosePRE2017, SwendsenPRR2018} for recent research on the correct definition of equilibrium temperature). 

The second case is given by \emph{small open} systems, which are in \emph{weak} contact with a large thermal bath. 
Then, von Neumann entropy (or its classical counterpart, the Gibbs-Shannon entropy) coincides with thermodynamic 
entropy even \emph{out of equilibrium}. This became consensus in classical stochastic~\cite{SekimotoBook2010, 
SeifertRPP2012, VandenBroeckEspositoPhysA2015} and quantum thermodynamics~\cite{KosloffEntropy2013, 
VinjanampathyAndersCP2016}. It was subject to a direct experimental test~\cite{GavrilovChetriteBechhoeferPNAS2017} 
and experimental confirmations of Landauer's principle further support this hypothesis~\cite{OrlovEtAlJJAP2012, 
BerutEtAlNature2012, JunGavrilovBechhoeferPRL2014, SilvaEtAlPRSA2016, HongEtAlSciAdv2016, YanEtAlPRL2018}. 

\subsection{Boltzmann entropy}

The second well-known microscopic candidate for thermodynamic entropy is Boltzmann's entropy. To define it precisely, 
we consider a special case of later relevance. Let $H$ be the Hamiltonian of an isolated system, where we dropped 
any dependence on external parameters $\lambda_t$ for notational simplicity. We write the stationary Schr\"odinger 
equation as 
\begin{equation}\label{eq Schroedinger stationary}
 H|E_i,\ell_i\rangle = E_i|E_i,\ell_i\rangle,
\end{equation}
where $|E_i,\ell_i\rangle$ denotes an energy eigenstate with eigenenergy $E_i$ and $\ell_i$ labels possible exact 
degeneraries. Now, imagine an isolated system with many components such that its associated Hilbert space is extremely 
large. For all practical purposes, it is then impossible that a measurement of the energy is so precise that it 
yields a unique eigenenergy $E_i$. Instead, any measurement has a finite resolution or uncertainty $\delta$, which 
can be mathematically captured by a projector of the form 
\begin{equation}\label{eq projector microcanonical}
 \Pi_E \equiv \sum_{E_i\in[E,E+\delta)} \sum_{\ell_i} |E_i,\ell_i\rl E_i,\ell_i|.
\end{equation}
These projectors form a complete and orthogonal set $\{\Pi_E\}$, i.e., $\sum_E \Pi_E = 1$ (with $1$ the identity 
operator) and $\Pi_E\Pi_{E'} = \delta_{E,E'}\Pi_E$ (with $\delta_{E,E'}$ the Kronecker delta). 
This describes a \emph{coarse-grained measurement}. 

Now, if such a coarse-grained measurement yields outcome $E$, the Boltzmann entropy of the system is 
\begin{equation}\label{eq Boltzmann entropy}
 S_B(E) \equiv \ln V_E,
\end{equation}
where $V_E = \mbox{tr}\{\Pi_E\}$ is the rank of the projector, called in the following also a \emph{volume term}. 
Thus, the Boltzmann entropy counts all possible microstates compatible with the constraint of knowing $E$, 
and then takes the logarithm of it (remember that $k_B\equiv1$). 

Clearly, if information about further macrocopic variables is available, e.g., the particle number $N$, then the 
Boltzmann entropy becomes 
\begin{equation}\label{eq Boltzmann entropy 2}
 S_B(E,N,\dots) \equiv \ln V_{E,N,\dots},
\end{equation}
where $V_{E,N,\dots}$ counts all possible microstates compatible with the constraints $E$, $N$, etc. We remark that 
the precise definition of $V_{E,N,\dots}$ is subtle if the corresponding observables do not commute. However, for the 
majority of applications in macroscopic thermodynamics, the corresponding observables commute at least approximately. 

A distinctive feature of Boltzmann's entropy compared to the von Neumann entropy is that it is nonzero even 
for a \emph{pure} state $\rho = |\psi\rl\psi|$. For instance, if the pure state is confined to the energy shell 
$[E,E+\delta)$, i.e., $\Pi_E|\psi\rangle = |\psi\rangle$, one confirms that 
\begin{equation}
 S_B(E) = \ln V_E \neq S_\text{vN}[|\psi\rl\psi|] = 0. 
\end{equation}

Moreover, Boltzmann's concept easily allows to explain the second law, even without the need to introduce any notion of 
ensembles. For an isolated system it is much more probable to evolve from a region of small volume towards a region of 
large volume and to reside for long times in the region with the largest volume, which is identified with 
thermodynamic equilibrium. This explains the increase of entropy after lifting a constraint and the 
tendency to find systems in a state of maximum entropy. 

The power and simplicity of Boltzmann's concept is so appealing that many researchers have univocally adapted the idea 
to identify Boltzmann's entropy with the phenomenological thermodynamic entropy for macroscopic systems, even out of 
equilibrium. Perhaps surprisingly, also Jaynes was a proponent of it~\cite{JaynesInBook1988, 
JaynesInBook1989, JaynesInBook1992}. 

In his words, ``we feel quickly that [Eq.~(\ref{eq Boltzmann entropy 2})]  must be correct, because of the light that 
this throws on our problem. Suddenly, the mysteries evaporate; the meaning of Carnot's principle, the reason for the 
second law, and the justification for Gibbs' variational principle, all become obvious'' (stated below Eq.~(17) in 
Ref.~\cite{JaynesInBook1988}) and ``the above arguments make it clear that [...] any macrostate---equilibrium or 
nonequilibrium---has an entropy~[(\ref{eq Boltzmann entropy 2})]'' (stated above Eq.~(25) in 
Ref.~\cite{JaynesInBook1989}).

Indeed, it is easy to recognize that Boltzmann's entropy fits well Jaynes' epistemological view on the second law for 
two resons. First, for the computation of Boltzmann's entropy ``it is necessary to decide at the outset of a 
problem which macroscopic variables or degrees of freedom we shall measure and/or control''~\cite{JaynesInBook1992}, 
where the ``macrosopic variables'' in Jaynes' language are our observables $E$, $N$, etc. Second, if these observables 
are fixed to a given accuracy, then the state reflecting maximum ignorance about the situation (i.e., maximum entropy 
in the information theory sense), is given by the microcanonical ensemble. If only the energy $E$ is known, this 
microcanonical ensemble reads 
\begin{equation}\label{eq microcanonical}
 \omega(E) \equiv \frac{\Pi_E}{V_E},
\end{equation}
which satisfies $S_\text{vN}[\omega(E)] = \ln V_E = S_B(E)$. 

Albeit also favouring the Boltzmann entropy, the purely epistemological nature of the second law is denied in 
Refs.~\cite{LebowitzPhysToday1993, GoldsteinEtAlInBook2020} by pointing out that the flow of heat from hot to cold in 
macroscopic systems is a \emph{fact}, which does not depend on the observer's state of knowledge. That is to say, 
one expects the laws of thermodynamics to be generically true, either on a distant planet (about which we have no 
knowledge) or in an isolated many-body system (where we might be able to prepare pure states). 

Independent of the reader's opinion on that matter (even the present authors do not fully agree on it), 
we find it important to point out that also Boltzmann's approach 
faces deficiencies in light of current experiments. In fact, as stressed above, there is an agreement in favor of von 
Neumann's (or Shannon's) entropy for small systems in weak contact with a thermal bath. Since 
todays nanotechnologies allow to make very precise measurements on small systems, the volume term appearing in 
Boltzmann's entropy can be one and hence, it's logarithm is zero. Therefore, Boltzmann entropy seems to be inadequate 
to take into account microscopic information, which is available to us now, but was not available a hundred years ago. 

To conclude, whereas von Neumann entropy appears \emph{too fine-grained} for all systems, which have more than 
a few degrees of freedom, Boltzmann's entropy appears \emph{too coarse-grained} to account for today's 
experimental capabilities. This is also once more exemplified in Fig.~\ref{fig entropy comparison}. It therefore seems 
desirable to have a flexible concept for entropy, which can interpolate between these two ideas. 

\begin{figure}
 \centering\includegraphics[width=0.49\textwidth,clip=true]{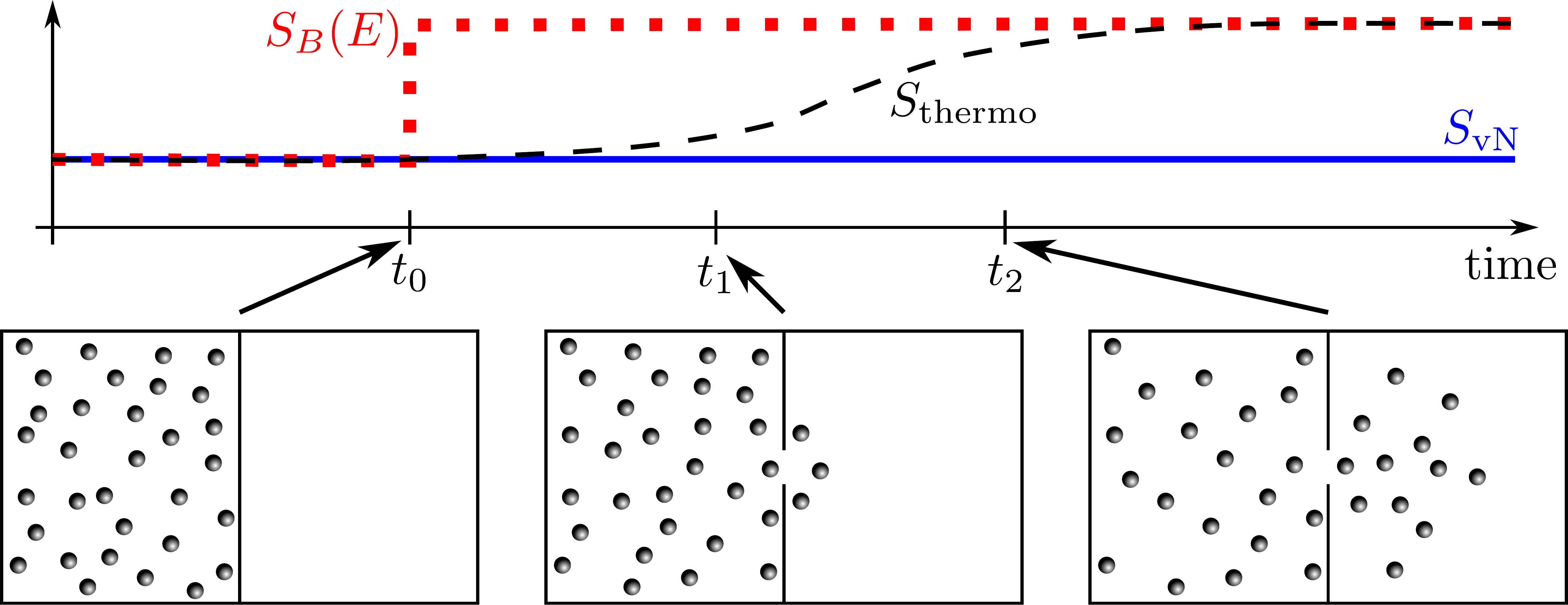}
 \label{fig entropy comparison} 
 \caption{Thought experiment with a gas in a box with perfectly reflecting boundaries. Initially, the gas is confined 
 to the left half of the box. Then, at $t_0$ a small hole is opened in the wall such that gas particles can diffuse 
 to the right. Since the gas needs time to diffuse, it seems sensible to demand that thermodynamic entropy should 
 smoothly interpolate between the lower initial and higher final value (dashed line). In contrast to this desideratum, 
 von Neumann entropy stays constant for all times (thick blue line). A naive application of Boltzmann's entropy 
 $S_B(E)$ captures the correct initial and final value, but contains a sudden discontinous jump at $t_0$ even 
 if the hole in the wall is very small. Thus, it misses some relevant dynamical information. For a similar and more 
 detailed discussion see Ref.~\cite{SafranekAguirreDeutschPRE2020}. }
\end{figure}

\subsection{Observational entropy}
\label{sec sub obs ent}

We now review a third concept, which is called observational (or coarse-grained) entropy and which overcomes the 
problems mentioned above. We begin with its formal definition, followed by the recapitulation of some useful 
mathematical facts needed later on, and we end with remarks about its appearance in the literature. 

\subsubsection{Formal definition}

We ignore any thermodynamic considerations for the moment and consider some \emph{coarse-graining} $X = \{\Pi_x\}$ 
defined by a complete set of orthogonal projectors satisfying $\sum_x \Pi_x = 1$ and 
$\Pi_x\Pi_{x'} = \delta_{x,x'}\Pi_x$. This coarse-graining can be associated to a measurement of a suitable 
observable, but the eigenvalues of the observable are unimportant for us. Instead, if the system is in state $\rho$, 
we only need the probability $p_x = \mbox{tr}\{\Pi_x\rho\}$ to observe outcome $x$ and the volume term 
$V_x \equiv \mbox{tr}\{\Pi_x\}$. Then, observational entropy with respect to a coarse-graining $X$ is defined as 
\begin{equation}\label{eq def obs ent}
 S_\text{obs}^X(\rho) \equiv \sum_x p_x (-\ln p_x + \ln V_x).
\end{equation}

To convince ourselves that observational entropy interpolates between the notions of von Neumann and Boltzmann 
entropy, we consider the following two cases. 

First, assume that we are an observer who, in von Neumann's words~\cite{VonNeumann1929, VonNeumannEPJH2010}, 
``can carry out all measurements that are possible in principle.'' Then, we could choose a coarse-graining 
$X = \{|x\rl x|\}$, which matches the eigenbasis of the state $\rho = \sum_x \lambda _x |x\rl x|$. We immediately 
reveal that in this case $S_\text{obs}^X(\rho) = S_\text{vN}(\rho)$. 

Second, observe that we can write observational entropy as 
\begin{equation}
 S_\text{obs}^X(\rho) = S_\text{Sh}(p_x) + \sum_x p_x S_B(x),
\end{equation}
where $S_\text{Sh}(p_x)$ is the Shannon entropy of the probabilities $p_x$ and the second term presents an averaged 
Boltzmann entropy. Thus, if $p_x = \delta_{x,x'}$, i.e., we are certain that the state $\rho$ is 
confined in the `macrostate' $\Pi_{x'}$, we obtain $S_\text{obs}^X(\rho) = S_B(x')$. Depending on the coarse-graining 
$X$, this allows us to reproduce, e.g., Boltzmann's entropy~(\ref{eq Boltzmann entropy}) 
associated to an imprecise energy measurement. 

The last point about the correct choice of coarse-graining is very important. Definition~(\ref{eq def obs ent}) 
formally holds for \emph{any} coarse-graining. To connect observational entropy to \emph{thermodynamic} entropy, we 
need to make the right choice of coarse-graining just as Jaynes indicated by saying that ``it is necessary to decide 
at the outset of a problem which macroscopic variables or degrees of freedom we shall measure and/or 
control''~\cite{JaynesInBook1992}. The difference is that observational entropy is not restricted to 
``macroscopic variables [= coarse-grainings],'' but can take into account more detailed information as well. 
Note that the correct choice of variables (energy, polarization, particle number, etc.) is often determined by 
the physical setup itself, whereas the level of `coarse-grainedness' quantified by the volume terms $V_x$ depends 
experimentally on the \emph{precision} of the measurement. Theoretically, this precision is a free parameter in 
principle, which has to be chosen reasonably. Luckily, one also expects that the qualitative picture does not 
sensitively depend on the precise choice of $V_x$. 

\subsubsection{Some elementary mathematical properties}

We now list a couple of mathematical facts as lemmas, which appear scattered throughout the literature (see the 
next subsection for a list of references). These lemmas add further appeal to the definition of observational 
entropy. They hold for any coarse-graining $\{\Pi_x\}$ and therefore might be of interest even outside thermodynamic 
considerations. Below, we also make use of the quantum relative entropy 
$D(\rho\|\sigma) \equiv \mbox{tr}\{\rho(\ln\rho-\ln\sigma)\}$, see 
Appendix~\ref{sec basics information theory} for further details. 

First, observational entropy can be bounded from above and below. 

\begin{lemma}
 If $\dim\C H < \infty$ denotes the dimension of the Hilbert space of the isolated system, then 
 \begin{equation}
 S_\text{\normalfont vN}(\rho) \le S_\text{\normalfont obs}^X(\rho) \le \ln\dim\C H. 
\end{equation}
\end{lemma}

Second, observational entropy is extensive in the limit where one expects it to be extensive. 

\begin{lemma}
 Consider a composite system in the decorrelated state $\rho = \rho_1\otimes\dots\otimes\rho_n$ and a composite 
 coarse-graining $X = X_1\otimes\dots\otimes X_n$ with projectors $\Pi_{x_1}\otimes\dots\otimes\Pi_{x_n}$. Then, 
 \begin{equation}\label{eq extensivity}
  S_\text{\normalfont obs}^X(\rho) = \sum_{j=1}^n S_\text{\normalfont obs}^{X_j}(\rho_j).
 \end{equation}
\end{lemma}

Of course, one expects Eq.~(\ref{eq extensivity}) to remain approximately true for weakly correlated system, which 
describe multiple macroscopic systems in contact with each other. If surface properties are negligible compared to 
their bulk properties, this then implies the usual notion of extensivity known from thermodynamics. 

Next, we note a useful rewriting of observational entropy. For that purpose, we introduce the notation 
$\rho(x) \equiv \Pi_x\rho\Pi_x/p_x$, which describes the post-measurement state given outcome $x$, 
and $\omega(x) \equiv \Pi_x/V_x$, which denotes a generalized `microcanonical ensemble' given the constraint $x$. 

\begin{lemma}
 We have 
 \begin{equation}\label{eq lemma 1}
  S_\text{\normalfont obs}^X(\rho) = 
  S_\text{\normalfont vN}\left[\sum_x p_x\rho(x)\right] + \sum_x p_x D[\rho(x)\|\omega(x)].
 \end{equation}
\end{lemma}

\begin{proof}
 Since the states $\rho(x)$ have support on orthogonal subspaces, it follows from Theorem~11.10 in 
 Ref.~\cite{NielsenChuangBook2000} that 
 \begin{equation}
  S_\text{vN}\left[\sum_x p_x\rho(x)\right] = \sum_x p_x \left\{S_\text{vN}[\rho(x)] - \ln p_x\right\}. 
 \end{equation}
 Using this insight in Eq.~(\ref{eq lemma 1}) yields 
 \begin{equation}
  \begin{split}
   S_\text{obs}^X(\rho) &= -\sum_x p_x\left[\ln p_x + \mbox{tr}\{\rho(x)\ln\omega(x)\}\right] \\
                        &= -\sum_x p_x\left[\ln p_x - \mbox{tr}\{\rho(x)\ln V_x\}\right],
  \end{split}
 \end{equation}
 which is identical to Eq.~(\ref{eq def obs ent}) since $\mbox{tr}\{\rho(x)\} = 1$. 
\end{proof}

The next lemma characterizes the states $\rho$ which have the same von Neumann and observational 
entropy. 

\begin{lemma}\label{lemma S obs equal S vN}
 We have $S_\text{\normalfont vN}(\rho) = S_\text{\normalfont obs}^X(\rho)$ if and only if 
 \begin{equation}\label{eq lemma 2}
  \rho = \sum_x p_x \omega(x)
 \end{equation}
 for an arbitary set of probabilities $p_x$. 
\end{lemma}

\begin{proof}
 Using Eq.~(\ref{eq lemma 1}), we can write 
 \begin{equation}
  \begin{split}
   S_\text{obs}^X(\rho) - S_\text{vN}(\rho) 
   =&~ S_\text{\normalfont vN}\left[\sum_x p_x\rho(x)\right] - S_\text{vN}(\rho) \\
   &+ \sum_x p_x D[\rho(x)\|\omega(x)].
  \end{split}
 \end{equation}
 We see that $S_\text{obs}^X(\rho) - S_\text{vN}(\rho)$ is given as the sum of two non-negative terms because it 
 follows from Theorem~11.9 in Ref.~\cite{NielsenChuangBook2000} that 
 \begin{equation}\label{eq thm 11 9}
  S_\text{\normalfont vN}\left[\sum_x p_x\rho(x)\right] - S_\text{vN}(\rho) \ge 0
 \end{equation}
 with equality if and only if $\rho = \sum_x p_x\rho(x)$. Furthermore, $D[\rho(x)\|\omega(x)] = 0$ if and only if 
 $\rho(x) = \omega(x)$. Hence, Eq.~(\ref{eq lemma 2}) follows. 
\end{proof}

The next lemma can be seen as a precursor of the second law, albeit the coarse-graining $\{\Pi_x\}$ is still arbitrary 
and not necessarily of thermodynamic relevance. It applies to any isolated system with 
time evolution~(\ref{eq time evo isolated}). Since we are now interested in 
\emph{changes} in observational entropy, we write $S_\text{obs}^{X_t}(t) = S_\text{obs}^{X_t}[\rho(t)]$ for the 
observational entropy at time $t$ and indicate that also the chosen coarse-graining $X = X_t$ can depend on time. 

\begin{lemma}\label{lemma 2nd law}
 If $S_\text{\normalfont obs}^{X_0}(0) = S_\text{\normalfont vN}[\rho(0)]$, then 
 \begin{equation}\label{eq 2nd law formal}
  \Delta S_\text{\normalfont obs}^{X_t}(t) = S_\text{\normalfont obs}^{X_t}(t) - S_\text{\normalfont obs}^{X_0}(0) \ge 0. 
 \end{equation}
\end{lemma}

\begin{proof}
 From the unitarity of time evolution we deduce 
 $S_\text{\normalfont vN}[\rho(0)] = S_\text{\normalfont vN}[\rho(t)]$ and hence, 
 \begin{equation}
  \Delta S_\text{obs}^{X_t}(t) = S_\text{obs}^{X_t}(t) - S_\text{\normalfont vN}[\rho(t)].
 \end{equation}
 This term is easily shown to be positive using Eqs.~(\ref{eq lemma 1}) and~(\ref{eq thm 11 9}). 
\end{proof}

\subsubsection{Historical remarks}

Definition~(\ref{eq def obs ent}) or, more often, similar but less general forms of it appear scattered 
throughout the literature on statistical mechanics. Not seldomly Eq.~(\ref{eq def obs ent}) is used in various 
computations without explicitly identifying it with thermodynamic entropy, in particular not out of 
equilibrium. Our efforts to trace back the origin and use of definition~(\ref{eq def obs ent}) has yielded the 
following results, which shall not imply that the given list is exhaustive. 

For classical systems, where one needs to coarse-grain the phase space into cells, variants of 
Eq.~(\ref{eq def obs ent}) appear already in the work of Gibbs~\cite{GibbsBook1902} and Lorentz~\cite{Lorentz1906}, 
see also Sec.~23a of the treatise about statistical mechanics of the Ehrenfests~\cite{EhrenfestEhrenfest1911}. In 
this context, Eq.~(\ref{eq def obs ent}) is also known as ``coarse-grained entropy'' (see Wehrl~\cite{WehrlRMP1978}, 
who connects it to ergodicity and mixing and cites further references). 

For quantum systems, Eq.~(\ref{eq def obs ent}) can be traced back to von Neumann, who attributes it to a personal 
communication from Wigner and clearly acknowledges its usefulness for problems in statistical 
mechanics~\cite{VonNeumann1929}. In fact, von Neumann proves a remarkable `$H$-theorem' in his work, which we 
summarize here informally (see also the accompanying article~\cite{GoldsteinEtAlEPJH2010} of the English 
translation~\cite{VonNeumannEPJH2010} for further details). For this purpose consider an isolated system 
with time-independent Hamiltonian $H$ and suppose that the Hamiltonian has no degenerate energy gaps. Furthermore, 
the orientation between the eigenvectors of $H$ and the eigenvectors of the coarse-graining $\{\Pi_x\}$ is assumed 
random. Finally, it is assumed that the projectors $\Pi_x$ are sufficiently `coarse,' i.e., the number of elements 
in the set $\{\Pi_x\}$ must be much smaller than $\dim\C H$. Then, von Neumann found that \emph{for all initial 
states $|\psi_0\rangle$ the observational entropy $S_\text{\normalfont obs}^X(|\psi_t\rl\psi_t|)$ will be close to 
its upper bound $\ln\dim\C H$ for most times $t$.} In particular, if the initial observational entropy 
$S_\text{obs}^X(|\psi_0\rl\psi_0|)$ is small, this proves the increase of entropy after waiting a sufficient amount 
of time. Von Neumann's $H$-theorem can 
be regarded as complementary to what we establish below: Whereas we rely on a special class of initial states, no 
assumption about the Hamiltonian $H$ is used. Finally, we remark that von Neumann's approach was recently 
extended~\cite{GoldsteinEtAlPRSA2010, RigolSrednickiPRL2012, ReimannPRL2015}, but the focus was on equilibration 
of expectation values, whereas his $H$-theorem related to observational entropy found no further attention. 

We continue by pointing out that a second-law-like increase of observational entropy similar to 
Lemma~\ref{lemma 2nd law} was proven for quantum systems in \S106 of Tolman's book~\cite{TolmanBook1938} and in 
Sec.~1.3.1 of the book by Zubarev \emph{et al.}~\cite{ZubarevMorozovRoepkeBook} (the proof in the classical case 
seems to date back to Gibbs, see again the Ehrenfests~\cite{EhrenfestEhrenfest1911}). It is, however, interesting to 
note that both books refuse to use Eq.~(\ref{eq def obs ent}) as a definition of thermodynamic entropy for 
out-of-equilibrium processes: Tolman discusses the connection to thermodynamic entropy only at equilibrium and prefers 
to use the Gibbs-Shannon-von Neumann entropy [compare with Eq.~(122.10) therein] and Zubarev \emph{et al.}~prefer 
the Gibbs-Shannon-von Neumann entropy of a generalized out-of-equilibrium Gibbs ensemble. Other sources, where 
observational entropy was sometimes more and sometimes less clearly identified as thermodynamic entropy, are 
Refs.~\cite{Pauli1928, PercivalJMP1961, PenroseRPP1979, LatoraBarangerPRL1999, NauenbergAJP2004, VanKampenBook2007, 
GemmerSteinigewegPRE2014, LeePRE2018}. 

The present tutorial was in particular inspired by the recent work of 
\u{S}afr\'anek, Deutsch and Aguirre~\cite{SafranekDeutschAguirrePRA2019a}, who coined the terminology ``observational 
entropy'' and propose it as a generally valid definition of thermodynamic entropy for isolated nonequilibrium quantum 
systems. Further arguments for it are also given in their subsequent work~\cite{SafranekDeutschAguirrePRA2019b, 
SafranekAguirreDeutschPRE2020, FaiezEtAlPRA2020, SchindlerSafranekAguirrePRA2020}, where also the case of 
multiple non-commuting coarse-grainings is treated. In our exposition, we only deal with single or multiple but 
commuting coarse-grainings. 

\section{Second law and effective temperature in isolated systems}
\label{sec entropy second law isolated}

We now consider the first thermodynamic application of observational entropy in isolated systems, thereby introducing 
concepts that turn out to be important for the open system paradigm studied in Sec.~\ref{sec first second law open}. 
For simplicity, we focus on a homogeneous isolated system with energy as the only relevant macrovariable. We beginn 
by studying entropy production in general followed by a discussion of the reversible case. We then introduce the 
important concept of an effective nonequilibrium temperature. Finally, we briefly discuss possible extensions. 

\subsection{Entropy production in a homogeneous system}
\label{sec sub EP homogenous}

We consider a driven isolated system with Hamiltonian $H(\lambda_t)$ and imagine the total energy of the isolated 
system to be the only relevant (or accessible) thermodynamic quantity. We call such a system \emph{homogenous} as 
we ignore any spatial irregularities. Thus, our coarse-graining is defined by $\{\Pi_{E_t}\}$, where $\Pi_{E_t}$ is 
obtained from the previously introduced projector~(\ref{eq projector microcanonical}) by replacing the eigenenergies 
$E_i$ and eigenstates $|E_i,\ell_i\rangle$ by $E_i(\lambda_t)$ and $|E_i(\lambda_t),\ell_i(\lambda_t)\rangle$ to take 
into account the external driving. 

The time evolution of the system is described by Eq.~(\ref{eq time evo isolated}). Denoting 
$p_{E_t}(t) = \mbox{tr}\{\Pi_{E(\lambda_t)}\rho(t)\}$ and $V_{E_t} = \mbox{tr}\{\Pi_{E(\lambda_t)}\}$, the 
observational entropy reads 
\begin{equation}\label{eq obs ent isolated system}
 S_\text{obs}^{E_t}[\rho(t)] \equiv \sum_{E_t} p_{E_t}(t)[-\ln p_{E_t}(t) + \ln V_{E_t}],
\end{equation}
which we use as our microscopic definition of thermodynamic entropy in this section. 

Next, we consider the set of states $\rho(t)$ that satisfy $S_\text{obs}^{E_t}[\rho(t)] = S_\text{vN}[\rho(t)]$. From 
Lemma~\ref{lemma S obs equal S vN} we know that this set is
\begin{equation}\label{eq equilibrium states}
 \Omega(\lambda_t) = \left\{\left.\sum_{E_t} p_{E_t} \omega(E_t)\right| p_{E_t} \text{ arbitrary}\right\}.
\end{equation}
These states correspond to a somewhat larger set of equilibrium states than conventionally considered in statistical 
mechanics, but they share the same feature: they are invariant in time for a \emph{fixed} Hamiltonian $H(\lambda_t)$ 
and, given a distribution $p_{E_t}$, they maximize the von Neumann entropy as a measure about our `ignorance' of the 
state. 

Moreover, whenever we start with a state $\rho(0)\in\Omega(\lambda_0)$, the second law follows 
from Lemma~\ref{lemma 2nd law}, 
\begin{equation}\label{eq EP isolated homo}
 \Sigma(t) = S_\text{obs}^{E_t}[\rho(t)] - S_\text{obs}^{E_0}[\rho(0)] \ge 0,
\end{equation}
independent of the unitary time evolution operator. Equation~(\ref{eq EP isolated homo}) is the entropy production 
of an isolated homogenous system for an energy coarse-graining. 

\subsection{Reversible case}
\label{sec sub reversible}

It is instructive to consider the reversible case of Eq.~(\ref{eq EP isolated homo}), defined by: 
$S_\text{obs}^{E_t}[\rho(t)] = S_\text{obs}^{E_0}[\rho(0)]$. While we are mostly interested in nonequilibrium 
situations in this article, the reversible case is an important limiting case in thermodynamics and typically 
(approximately) generated by changing the protocol $\lambda_t$ very slowly. 

The goal of this section is to prove that reversible processes are characterized by the fact that they are easy to 
time-reverse from a macroscopic point of view, in unison with our knowledge from thermodynamics. For that purpose, we 
need the notion of time-reversal symmetry, which is introduced in greater detail in Appendix~\ref{sec time reversal}. 

We recall how to time-reverse a unitary process \emph{in principle}. Let $\rho(t) = U(t,0)\rho(0)U^\dagger(t,0)$ be 
the time evolved state in the `forward process.' We denote by $\Theta$ the anti-unitary time-reversal operator. 
Consequently, $U_\Theta(t,0) = \Theta U^\dagger(t,0)\Theta^{-1}$ becomes the unitary time evolution operator 
generated by the Hamiltonian $H_\Theta(\lambda_s,B) = H(\lambda_{t-s},-B)$ with a time-reversed driving protocol and 
a reversed magnetic field $B$. Finally, let $\Theta\rho(t)\Theta^{-1}$ denote the time-reversed final state of the 
forward process. Then, time-reversal symmetry guarantees that we can recover the initial state $\rho(0)$ by 
\begin{equation}\label{eq time reversed process}
 \rho(0) = \Theta^{-1} U_\Theta(t,0)\Theta\rho(t)\Theta^{-1} U_\Theta^\dagger(t,0)\Theta.
\end{equation}
In words, we recover the initial state of the forward process if we time-reverse the final state, let the 
protocol run backward (and perhaps reverse a magnetic field), and apply the inverse time-reversal on the state. 
Here, the experimentally easy part is to reverse the driving protocol and a magnetic field. The hard part 
instead corresponds to time-reversing the state $\rho(t)$ (for instance, classically this requires to flip all momenta 
$p\rightarrow-p$, which already for a single particle is hard to achieve accurately). Moreover, since $\Theta$ is 
anti-unitary, it can not be implemented in a lab in general. An implementation of Eq.~(\ref{eq time reversed process}) 
therefore remains experimentally out of reach in most cases. 

There is, however, one important class of exceptions: the operation $\Theta\rho(t)\Theta^{-1}$ is easy to achieve if 
the states $\rho(t)$ are symmetric under time-reversal. These states are precisely the set of states 
characterized by Eq.~(\ref{eq equilibrium states}). Symbolically, we can denote this by 
$\Theta\Omega(\lambda_t)\Theta^{-1} = \Omega(\lambda_t)$ or 
$\Theta\Omega(\lambda_t,B)\Theta^{-1} = \Omega(\lambda_t,-B)$ in 
presence of a magnetic field. In words, an equilibrium state is invariant under time-reversal and hence, there is 
actually no need to implement the cumbersome time-reversal operation (apart from perhaps flipping $B$). 

Now, we return to the reversible case of Eq.~(\ref{eq EP isolated homo}). From 
$S_\text{obs}^{E_t}[\rho(t)] = S_\text{obs}^{E_0}[\rho(0)]$ and 
$S_\text{obs}^{E_0}[\rho(0)] = S_\text{vN}[\rho(0)] = S_\text{vN}[\rho(t)]$ we can conlude 
(cf.~Lemmas~\ref{lemma S obs equal S vN}) that also the final state must be an equilibrium 
state: $\rho(t) \in \Omega(\lambda_t)$. Thus, our approach based on observational entropy shows that \emph{reversible 
processes} are characterized by the fact that they are \emph{simple to time-reverse from a macrocopic point of view}. 

This statement is not a trivial tautology. If we had started with a different entropy concept, it is unclear 
whether this would imply the same statement. For instance, if we had identified the von Neumann entropy with 
thermodynamic entropy, then we would need to call all processes `reversible' despite being very complicated to 
time-reverse in practice. 

\subsection{Effective nonequilibrium temperature}
\label{sec sub noneq temp}

Up to now, we have introduced microscopic notions for internal energy, mechanical work and thermodynamic entropy. 
Another important quantity is \emph{temperature}. This concept also plays an important role in the next section, 
where our goal is to provide a microscopic derivation of the phenomenological Clausius 
inequality~(\ref{eq EP traditional}), which remains valid even \emph{out of equilibrium}. 

We remark that the definition of meaningful nonequilibrium temperatures has a long 
history~\cite{CasasVazquezJouRPP2003}. The definition we adapt here has appeared in phenomenological nonequilibrium 
thermodynamics under the name ``nonequilibrium contact temperature'' more than 40 years ago~\cite{Muschik1977, 
MuschikBrunkIJES1977}. It has also appeared at various places in the statistical mechanics literature (see, e.g., 
Refs.~\cite{TasakiArXiv2000, RitortJPCB2005, JohalPRE2009, SeifertPA2020}) without, however, enjoying a wider 
popularity. 

For an arbitrary nonequilibrium state $\rho(t)$ we define the inverse nonequilibrium temperature $\beta^*_t$ by 
demanding  
\begin{equation}\label{eq def noneq temp}
 U(t) = \mbox{tr}\{H(\lambda_t)\rho(t)\} \equiv \mbox{tr}\{H(\lambda_t)\pi(\beta_t^*)\},
\end{equation}
i.e., we ask which inverse temperature does a fictitious Gibbs state $\pi(\beta_t^*)$ need to have such that its 
internal energy matches the true internal energy. In terms of our coarse-grained energy 
measurement~(\ref{eq projector microcanonical}), a Gibbs state is approximatively described by probabilities 
$\pi_{E_t}(\beta) \approx V_{E_t} e^{-\beta E_t}/\C Z(\beta,\lambda_t)$ with 
$\C Z(\beta,\lambda_t) = \sum_{E_t} V_{E_t} e^{-\beta E_t}$. Definition~(\ref{eq def noneq temp}) can then be 
also expressed as 
\begin{equation}\label{eq def noneq temp 2}
 \sum_{E_t} E_t p_{E_t}(t) \equiv \sum_{E_t} E_t \pi_{E_t}(\beta^*_t).
\end{equation}
Of course, definitions~(\ref{eq def noneq temp}) and~(\ref{eq def noneq temp 2}) match only if the measurement 
uncertainty (or thickness of the energy shell) $\delta$ is chosen sufficiently small such that 
$U(t) \approx \sum_{E_t} E_t p_{E_t}(t)$. For ease of presentation, we assume this in the following. 

An alternative way to describe the meaning of the nonequilibrium temperature $T^*_t = 1/\beta^*_t$ is as 
follows~\cite{Muschik1977, MuschikBrunkIJES1977}. Suppose that we have a collection of \emph{superbaths} at our 
disposal, which are prepared at different equilibrium temperatures $T$. Then, $T^*_t$ is defined to be 
the temperature $T$ of a superbath, which causes \emph{no net heat exchange} when coupling the system to it. 

Assuming the Hamiltonian to be time-independent for a moment, another property of $\beta_t^*$ follows by recalling 
that the canonical ensemble $\pi(\beta)$ with internal energy $\C U(\beta) = \mbox{tr}\{H\pi(\beta)\}$ satisfies 
\begin{equation}
 d\C U(\beta) = \C C(T) dT = -\frac{\C C(\beta)}{\beta^2} d\beta,
\end{equation}
where $d\C U(\beta) = \C U(\beta+d\beta) - \C U(\beta)$ and 
$\C C(\beta) = \beta^2 [\mbox{tr}\{H^2\pi(\beta)\} - \mbox{tr}\{H\pi(\beta)\}^2]$ denotes the heat capacity, which 
is \emph{non-negative}. Thus, by definition of the effective inverse temperature we can conclude that 
$\beta^* = \beta^*(U)$ is monotonically decreasing as a function of the internal energy $U$, stretching from 
$\beta^* = \infty$ if the system is in its ground state to $\beta^* = -\infty$ if the system is in its highest excited 
state (assuming the Hamiltonian is bounded from above, otherwise $\beta^*$ remains positive). 

Finally, $\beta^*$ allows us to establish a remarkable connection between energy and entropy, even out of equilibrium. 
Let ${\cal S}(\beta,\lambda) = S_\text{vN}[\pi(\beta,\lambda)]$ denote the von Neumann entropy of a Gibbs state at inverse temperature $\beta$. It follows that 
$T_t^*d\C S(\beta_t,\lambda_t) = dU(t) - \mbox{tr}\{[dH(\lambda_t)]\pi_t(\beta_t^*)\}$. Here, $dU(t)$ is the change 
of the nonequilibrium internal energy in Eq.~(\ref{eq def noneq temp}), and the term 
$\mbox{tr}\{[dH(\lambda_t)]\pi_t(\beta_t^*)\}$ can be interpreted as the work done on the system 
during a (fictitious) equilibrium process. Let us label $T_t^*d\C S(\beta_t,\lambda_t) \equiv \dbar Q(t)$ as a heat 
flux for reasons that will become clear in the next section. Then, 
\begin{equation}\label{eq entropy Clausius term}
 \C S(\beta_t^*) - \C S(\beta_0^*) = \int \frac{\dbar Q(s)}{T_s^*}
\end{equation}

Thus, the entropy production~(\ref{eq EP isolated homo}) can be written as: 
\begin{equation}
 \begin{split}\label{eq EP isolated Clausius}
  \Sigma(t) =&~ S_\text{obs}^{E_t}[\rho(t)] - \C S(\beta_t^*,\lambda_t) + \int \frac{\dbar Q(s)}{T_s^*} \\
            &+ \C S(\beta_0^*,\lambda_0) - S_\text{obs}^{E_0}[\rho(0)]
 \end{split}
\end{equation}
In particular, if the isolated system is prepared in a Gibbs state, the last line vanishes. Furthermore, since the 
Gibbs state maximizes entropy with respect to a fixed energy, we can conclude 
$S_\text{obs}^{E_t}[\rho(t)] \le S(\beta_t^*,\lambda_t)$. Consequently, 
\begin{equation}
 \int \frac{\dbar Q(s)}{T_s^*} \ge \Sigma(t) \ge 0,
\end{equation}
which we will use in the next section. 

\subsection{Extensions}
\label{sec sub extensions}

The simple description of an isolated system in terms of a homogenous coarse-grained energy variable 
covers only a small fraction of thermodynamically interesting situations. For instance, an accurate description 
of ultracold atoms experiments~\cite{BlochDalibardZwergerRMP2008, LewensteinSanperaAhufingerBook2012} likely 
requires further variables (particle number, magnetization, polarization, etc.) and, perhaps, variables with spatial 
resolution (e.g., energy or particle \emph{densities}). Since it seems impossible to cover all these experiments in a 
tutorial article, we focused only on the basics above. They can be generalized by refining the coarse-graining 
and illustrative examples have been already investigated by using observational 
entropy~\cite{SafranekDeutschAguirrePRA2019a, SafranekDeutschAguirrePRA2019b, FaiezEtAlPRA2020}. 

\section{First and second law in open systems}
\label{sec first second law open}

In this section, we derive the hierarchy of second 
laws~(\ref{eq EP basic basic}),~(\ref{eq EP basic}),~(\ref{eq EP traditional}) and~(\ref{eq EP traditional OQS}) 
for a suitable coarse-graining reflecting the degree of control an external agent has about the open quantum system. 
To derive Clausius' inequality~(\ref{eq EP traditional}) we use the microscopic definition~(\ref{eq def noneq temp}) 
of temperature and introduce definitions for heat and internal energy of an \emph{open} system. The present 
treatment presents a significant extensions of earlier work using observational entropy~\cite{StrasbergArXiv2019b}. 
Further generalizations of this approach (initially correlated states, multiple baths, particle transport) are treated 
in Sec.~\ref{sec further extensions}. 

\subsection{Relevant coarse-graining and initial state}
\label{sec sub relevant coarse graining}

The central idea of system-bath theories is to divide the universe into relevant degrees of freedom (the `system') and 
irrelevant degrees of freedom (the `bath')~\cite{BreuerPetruccioneBook2002, DeVegaAlonsoRMP2017}. The relevant 
degrees of freedom are assumed to be accessible by experiment, whereas only limited information is available 
about the irrelevant degrees of freedom. Many current experimental platforms---such as cavity or ciruit QED setups, 
optomechanical or nanoelectromechanical systems, quantum dots or nitrogen vacancy (NV) centers---show such a separation 
between precisely measurable system quantities and coarse information about the bath. 

Our definition of observational entropy is supposed to reflect this situation and, therefore, we choose the 
coarse-graining $\{|s\rl s|\otimes \Pi_{E_B}\}$. Here, $\{|s\rl s|\}$ is a set of rank-1 projectors acting on the 
system Hilbert space, whereas $\{\Pi_{E_B}\}$ is a set of coarse-grained energy projectors for the bath, i.e., 
$\Pi_{E_B}$ is constructed as in Eq.~(\ref{eq projector microcanonical}), but with respect to the bath Hamiltonian 
$H_B$. Thus, a measurement yielding outcome $(s,E_B)$ with probability 
$p_{s,E_B} = \mbox{tr}_{SB}\{|s\rl s|\otimes \Pi_{E_B}\rho_{SB}\}$ gives us complete knowledge about the microstate 
of the system, but reveals only partial information about the energy of the bath (which is related to its 
temperature). We remark that the basis for the system coarse-graining is arbitrary and might change in time. 
Therefore, we write $|s\rangle = |s_t\rangle$ in the following. Then, the observational entropy follows as 
\begin{equation}\label{eq obs ent OQS 1 bath}
 S_\text{obs}^{S_t,E_B}[\rho_{SB}(t)] \equiv -\sum_{s_t,E_B} p_{s_t,E_B}(t)\ln\frac{p_{s_t,E_B}(t)}{V_{E_B}},
\end{equation}
where $V_{E_B} = \mbox{tr}_B\{\Pi_{E_B}\}$ counts the number of bath microstates compatible with outcome $E_B$. 
Equation~(\ref{eq obs ent OQS 1 bath}) is our microscopic definition for thermodynamic entropy in the following. 

We believe that the coarse-graining above most accurately reflects the current spirit of open quantum system 
theory and many nanotechnological platforms. We remark, however, that all the following identities---unless 
otherwise stated---are valid for a system and a bath of any size. Moreover, by choosing a coarse-graining for the 
bath as in Sec.~\ref{sec entropy second law isolated}, we implicitly assumed the bath to be a homogeneous object from 
a macroscopic point of view. Perhaps not too far in the future, it might be necessary to extend the present approach 
to take into account further information about the bath in form of, e.g., spatial irregularities. Furthermore, if the 
system itself becomes large (say, larger than 10 qubits), the description in terms of fine-grained rank-1 projectors 
$|s_t\rl s_t|$ might no longer be adequate. Whatever information is necessary to accurately describe the experiment, 
the present approach can be adapted accordingly. 

Finally, we fix the initial state. In unison with the conventional open quantum systems 
approach~\cite{BreuerPetruccioneBook2002, DeVegaAlonsoRMP2017}, we consider in this section an initial state of the 
form 
\begin{equation}\label{eq initial state OQS}
 \rho_{SB}(0) = \rho_S(0) \otimes \pi_B(\beta_0).
\end{equation}
This describes a system state $\rho_S(0)$ initially decorrelated from a bath described by a Gibbs state at inverse 
temperature $\beta_0$. Since we are allowed to choose any $\{|s_0\rl s_0|\}$, we assume $\lr{s_0|\rho_S(0)|s'_0} = 0$ 
for all $s_0\neq s'_0$ in the following. Indeed, if the experimenter initially performs a measurement with projectors 
$\{|s_0\rl s_0|\}$, then this assumption holds automatically. Furthermore, as in Sec.~\ref{sec sub noneq temp}, we 
assume the resolution $\delta$ of the energy measurement of the bath to be sufficiently small such that 
\begin{equation}\label{eq delta agreement OQS}
 \C S_B(\beta_0) \equiv S_\text{vN}[\pi_B(\beta_0)] \approx S_\text{obs}^{E_B}[\pi_B(\beta_0)].
\end{equation}
This implies that we are consistent at equilibrium, with observational entropy coinciding with the standard 
equilibrium entropy of a canonical ensemble. 

Extensions of the initial state~(\ref{eq initial state OQS}) to take into account system-bath correlations or a bath 
not prepared in a canonical ensemble are treated in Sec.~\ref{sec further extensions}. 

\subsection{General second law}

Using the properties of the initial state discussed above, we confirm that 
$S_\text{vN}[\rho_{SB}(0)] = S_\text{obs}^{S_0,E_B}[\rho_{SB}(0)]$. Thus, from Lemma~\ref{lemma 2nd law} we directly 
find 
\begin{equation}\label{eq EP 2 micro}
 \Sigma_a(t) = S_\text{obs}^{S_t,E_B}[\rho_{SB}(t)] - S_\text{obs}^{S_0,E_B}[\rho_{SB}(0)] \ge 0.
\end{equation}
This quantifies the entropy production in our setup in its most general form. The subscript `a' on $\Sigma_a(t)$ 
shall remind us that this entropy production corresponds to the phenomenological second law of 
Eq.~(\ref{eq EP basic basic}). 

Furthermore, the decorrelated initial state~(\ref{eq initial state OQS}) implies 
$S_\text{obs}^{S_0,E_B}[\rho_{SB}(0)] = S_\text{obs}^{S_0}[\rho_S(0)] + S_\text{obs}^{E_B}[\rho_B(0)]$. At any later 
time we can write 
\begin{equation}
 \begin{split}
  S_\text{obs}^{S_t,E_B}[\rho_{SB}(t)] =&~ S_\text{obs}^{S_t}[\rho_S(t)] + S_\text{obs}^{E_B}[\rho_B(t)] \\
  &- I_\text{obs}^{S_t,E_B}[\rho_{SB}(t)].
 \end{split}
\end{equation}
Here, the classical mutual information (see Appendix~\ref{sec basics information theory}) 
\begin{equation}
 I_\text{obs}^{S_t,E_B}[\rho_{SB}(t)] = \sum_{s_t,E_B} p_{s_t,E_B}(t) \ln\frac{p_{s_t,E_B}(t)}{p_{s_t}(t)p_{E_B}(t)}
\end{equation}
characterizes the correlations in the final measurement result $(s_t,E_B)$. Since it is non-negative, we obtain 
\begin{equation}\label{eq EP 3 micro}
 \Sigma_b(t) = \Delta S_\text{obs}^{S_t}[\rho_S(t)] + \Delta S_\text{obs}^{E_B}[\rho_B(t)] \ge 0,
\end{equation}
which is the microscopic analogue of the phenomenological second law~(\ref{eq EP basic}). It follows that 
\begin{equation}
 \Sigma_b(t) - \Sigma_a(t) = I_\text{obs}^{S_t,E_B}[\rho_{SB}(t)] \ge 0. 
\end{equation}

The relevance of the mutual information term for the (thermo)dynamics of open quantum systems still needs 
further elucidation. In general, it obeys the inequalities 
\begin{equation}\label{eq MI inequalities}
 0 \le I_\text{obs}^{S_t,E_B}[\rho_{SB}(t)] \le I_{S:B}[\rho_{SB}(t)] \le 2\ln\dim\C H_S,
\end{equation}
where we assumed the Hilbert space dimension $\dim\C H_S$ of the system to be smaller than the Hilbert space 
dimension of the bath. Furthermore, numerical results~\cite{PtaszynskiEspositoPRL2019, 
RieraCampenySanperaStrasbergPRXQ2021} suggest that the mutual information can be large in view of these bounds: 
if $s_t$ denotes measurements of the system energy and if the system is undriven ($\lambda_t =$ constant), then---as a 
result of the microscopic conservation of energy---strong system-bath correlations can build up. 
But if the system is driven, correlations seem to diminish~\cite{RieraCampenySanperaStrasbergPRXQ2021} and the 
entropy production will be dominated by changes in the bath entropy $\Delta S_\text{obs}^{E_B}[\rho_B(t)]$, which can 
grow proportional with time $t$ in contrast to the mutual information~\cite{PtaszynskiEspositoPRL2019}. 

\subsection{Heat, internal energy and Clausius inequality}
\label{sec sub Clausius}

We now derive Clausius' inequality~(\ref{eq EP traditional}). It quantifies the entropy production of a system 
undergoing a nonequilibrium process while being in contact with a bath, whose temperature changes due to the flow of 
heat. 

Recall Sec.~\ref{sec sub noneq temp} where we defined a nonequilibrium temperature for any isolated system. 
This definitions also applies equally well to any subsystem. Thus, let $T_t^*$ denote the time-dependent 
nonequilibrium temperature of the bath, obtained from Eq.~(\ref{eq def noneq temp}) by adding a 
subscript $B$ to all quantities. Then, from Eqs.~(\ref{eq EP isolated Clausius}) and~(\ref{eq delta agreement OQS}) 
we infer that the change in bath entropy is 
\begin{align}\label{eq Clausius bath}
 \Delta S_\text{obs}^{E_B}[\rho_B(t)] 
 &= S_\text{obs}^{E_B}[\rho_B(t)] - \C S_B(\beta_t^*) + \int \frac{dU_B(s)}{T_s^*} \nonumber \\
 &\le \int \frac{dU_B(s)}{T_s^*}.
\end{align}
Here, we used that $d{\cal S}_B(\beta_s^*) = \beta_s^* dU_B(s)$ and $dU_B(s) = \dbar Q(s)$ because the bath 
Hamiltonian is time-independent. Furthermore, we used $S_\text{obs}^{E_B}[\rho_B(t)] \le {\cal S}_B(\beta_t^*)$ as also 
done below Eq.~(\ref{eq EP isolated Clausius}).

The idea to identify the change in internal energy of the bath with the heat flux appears very convincing at 
this point. Following our convention to count the energy flux into the system positive, we set 
$\dbar Q(s) \equiv -dU_B(s)$. It follows from Eqs.~(\ref{eq EP 3 micro}) and~(\ref{eq Clausius bath}) that 
\begin{equation}\label{eq EP 4 micro}
 \Sigma_c(t) = \Delta S_\text{obs}^{S_t}[\rho_S(t)] - \int \frac{\dbar Q(s)}{T_s^*} \ge 0.
\end{equation}
This constitutes a microscopic derivation of Clauisus' inequality~(\ref{eq EP traditional}). It extends an earlier 
analysis~\cite{JarzynskiJSP1999} by not assuming the bath to be at equilibrium at each time step. 

It is instructive to discuss the consequences of the identification $\dbar Q(s) \equiv -dU_B(s)$ further. First, 
one finds 
\begin{equation}
 \Sigma_c(t) - \Sigma_b(t) = \C S_B(\beta_t^*) - S_\text{obs}^{E_B}[\rho_B(t)] \ge 0.
\end{equation}
As expected, this difference is zero if the bath is also at later times well described by an equilibrium state with 
temperature $T_t^*$. In general, however, $\Sigma_c$ overestimates $\Sigma_b$ by neglecting potential 
nonequilibrium resources stored in the distribution of bath energies $E_B$ at time $t$. Such nonequilibrium resources 
were indeed recently studied in Refs.~\cite{SanchezSplettstoesserWhitneyPRL2019, HajilooEtAlPRB2020}. 

Second, the present identification of heat forces us, by virtue of the first law~(\ref{eq 1st law}), to identify the 
internal energy of the \emph{open} system as 
\begin{equation}\label{eq def internal energy open}
 U_S(t) \equiv \mbox{tr}_{SB}\{[H_S(\lambda_t) + V_{SB}]\rho_{SB}(t)\}. 
\end{equation}
In fact, based on this definition it is easy to microscopically verify the first law $\Delta U_S(t) = Q(t) + W(t)$. 
Here, $Q(t) = \int_0^t \dbar Q(s)$ is the total heat flow into the system and the mechanical work $W(t)$ was defined 
in Eq.~(\ref{eq def work and power}). We concluded in Sec.~\ref{sec energy work} that this definition of work is 
unambiguous for a driven system. Using the form of the system-bath Hamiltonian~(\ref{eq system bath Hamiltonian}), 
Eq.~(\ref{eq def work and power}) simplifes to 
\begin{equation}\label{eq def work OQS}
 W(t) = \int_0^t ds \mbox{tr}_S\left\{\frac{\partial H_S(\lambda_s)}{\partial s}\rho_S(s)\right\}.
\end{equation}

The definition~(\ref{eq def internal energy open}) seems to naturally follow in the present framework and it has also 
appeared in different earlier approaches~\cite{BassettPRA1978, LindbladBook1983, PeresBook2002, AndrieuxEtAlNJP2009, 
EspositoLindenbergVandenBroeckNJP2010, SagawaUedaPRL2010b, TakaraHasegawaDriebePLA2010}. It is, however, important 
to point out that Eq.~(\ref{eq def internal energy open}) can not be computed by only knowing the reduced system 
state $\rho_S(t)$ as it includes the interaction Hamiltonian $V_{SB}$, which is a disadvantage of the present 
definition. Only in the weak coupling regime, where the effect of $V_{SB}$ is assumed negligible compared to $H_S$ and 
$H_B$, we have $U_S(t) \approx \mbox{tr}_S\{H_S(\lambda_t)\rho_S(t)\}$. 

In fact, beyond weak coupling the correct definition of heat and internal energy is fiercly debated. Different  
proposals exist for quantum systems~\cite{LudovicoEtAlPRB2014, EspositoOchoaGalperinPRB2015, StrasbergEtAlNJP2016, 
BruchEtAlPRB2016, KatoTanimuraJCP2016, NewmanMintertNazirPRE2017, BeraEtAlNC2017, StrasbergEtAlPRB2018, 
LudovicoEtAlPRB2018, DouEtAlPRB2018, StrasbergEspositoPRE2019, RivasPRL2020} and also the classical case remains 
debated~\cite{SeifertPRL2016, TalknerHaenggiPRE2016, JarzynskiPRX2017, StrasbergEspositoPRE2017, MillerAndersPRE2017, 
StrasbergEspositoPRE2020, TalknerHaenggiRMP2020, TalknerHaenggiPRE2020, StrasbergEspositoPRE2020b}. 
The goal of this tutorial is \emph{not} to advertise Eq.~(\ref{eq def internal energy open}) as the only meaningful 
candidate. We believe, however, that the present framework helps to advance the debate for two reasons. 

First, we established a link between heat and entropy changes in the bath in Eq.~(\ref{eq Clausius bath}). 
It indicates that heat remains a meaningful concept even if the bath is not at equilibrium, but it no longer is the 
only contribution to the change in bath entropy. This important link has not been established in previous 
approaches. 

Second, our identification of heat results from our choice to view the coarse-grained energy of the bath as the 
relevant variable. We emphasized already that different, more refined choices are possible. This might in 
particular be relevant at strong coupling. Checking which of the many different proposals above can be explained by 
using observational entropy with respect to a \emph{different} coarse-graining would add further appeal and 
additional insights to them. 

\subsection{Weakly perturbed bath}

We complete the derivation of the laws of thermodynamics by deriving Eq.~(\ref{eq EP traditional OQS}). From the 
phenomenological description we expect Eq.~(\ref{eq EP traditional OQS}) to emerge out of the second 
law~(\ref{eq EP traditional}) whenever the bath can be approximated as \emph{static} such that its thermodynamic 
parameters do not change. This is often justified if the bath is very large and the system very small.

To reflect this idea in our framework, we start by expanding the probabilities $p_{E_B}(t)$ to measure the bath energy 
$E_B$ at time $t$ as 
\begin{equation}\label{eq bath state correction}
 p_{E_B}(t) = \pi_{E_B}(\beta_0)[1+\epsilon q_{E_B}(t)].
\end{equation}
Here, $\pi_{E_B}(\beta_0) = V_{E_B} e^{-\beta_0 E_B}/\C Z_B(\beta_0)$ denotes the initial probability to measure 
$E_B$ and $q_{E_B}(t) \in[-1,1]$ is a correction term, which, due to normalization, satisfies 
$\sum_{E_B} \pi_{E_B}(\beta_0)q_{E_B}(t) = 0$. A \emph{weakly perturbed bath} is now described by the situation where 
the parameter $\epsilon$ is small enough such that terms of order $\C O(\epsilon^2)$ are negligible. 

We now apply this idea to compute the change in bath entropy. By using Eq.~(\ref{eq bath state correction}), we get
\begin{equation}\label{eq relation heat}
 \Delta S_\text{obs}^{E_B}(t) = \beta_0\Delta U_B(t) + \C O(\epsilon^2),
\end{equation}
where the change in coarse-grained bath energy is 
\begin{equation}
 \begin{split}
  \Delta U_B(t) &= \sum_{E_B} E_B[p_{E_B}(t) - \pi_{E_B}(\beta_0)] \\ 
                &= \epsilon \sum_{E_B} E_B \pi_{E_B}(\beta_0)q_{E_B}(t).
 \end{split}
\end{equation}
Likewise, Eq.~(\ref{eq bath state correction}) also implies that the final nonequilibrium temperature must be 
$\epsilon$-close to the initial temperature: $|T_t^* - T_0| = \C O(\epsilon)$. We then obtain 
\begin{equation}
 \int_0^t \frac{d U_B(s)}{T_s^*}ds = \frac{\Delta U_B(t)}{T_0} + \C O(\epsilon^2)
\end{equation}
since $\Delta U_B(t)$ is itself of order $\epsilon$. 

Thus, for a weakly perturbed bath we can conclude 
\begin{equation}
 \Sigma_c(t) \approx \Sigma_d(t) = \Delta S_\text{obs}^{S_t}[\rho_S(t)] - \frac{Q(t)}{T_0} \ge 0. 
\end{equation}
This finishes our derivation of the hierarchy of second laws. Since the above inequality holds for all system 
coarse-grainings $\{|s_t\rl s_t|\}$, we can also choose it to coincide with the eigenbasis of $\rho_S(t)$. Then, 
we get $\Delta S_\text{obs}^{S_t}[\rho_S(t)] = \Delta S_\text{vN}[\rho_S(t)]$ and the second law becomes 
\begin{equation}\label{eq EP 5 micro}
 \Sigma_c(t) \approx \Sigma_d(t) = \Delta S_\text{vN}[\rho_S(t)] - \frac{Q(t)}{T_0} \ge 0. 
\end{equation}
This expression of the second law if often found in the context of open quantum system 
theory~\cite{BreuerPetruccioneBook2002, KosloffEntropy2013}. 
We conclude this section by putting our results in context of two other findings. 

First, Eq.~(\ref{eq EP 5 micro}) is often written for an infinitesimal time step as 
\begin{equation}\label{eq EP rate}
 \dot\Sigma_d(t) = \frac{d}{dt} S_\text{vN}[\rho_S(t)] - \frac{\dot Q(t)}{T_0},
\end{equation}
where $\dot\Sigma_d(t)$ is the entropy production \emph{rate} and $\dot Q(t) = -dU_B(t)/dt$. Whereas the 
non-negativity of Eq.~(\ref{eq EP 5 micro}) is guaranteed, the non-negativity of the entropy production rate 
$\dot\Sigma_d(t)$ is \emph{not}. However, one has $\dot\Sigma_d(t) \ge 0$ if the dynamics of the open system state 
$\rho_S(t)$ is described by the so-called Born-Markov-secular master equation, which has become---despite its many 
approximations involved---a widely used tool in the field~\cite{BreuerPetruccioneBook2002, KosloffEntropy2013, 
SchallerBook2014}. Similar approximations can be also used to derive a master equation for the probabilities 
$p_{s_t,E_B}(t)$. Then, in analogy to the previous case, one can confirm that 
$\dot\Sigma_a(t) = dS_\text{obs}^{S_t,E_B}[\rho_{SB}(t)]/dt \ge 0$~\cite{RieraCampenySanperaStrasbergPRXQ2021}. 
We remark that Markovianity alone is not sufficent to guarantee 
the non-negativity of the entropy production rate in general~\cite{StrasbergEspositoPRE2019}. 

Second, Eq.~(\ref{eq EP 5 micro}) emerged out of the more general version~(\ref{eq EP 4 micro}) of the second law for 
a weakly perturbed bath in unison with the phenomenological theory. Somewhat remarkably, it is 
possible to show that Eq.~(\ref{eq EP 5 micro}) always holds for the initial condition~(\ref{eq initial state OQS}), 
regardless of how far the bath is pushed away from equilibrium~\cite{LindbladBook1983, PeresBook2002, 
EspositoLindenbergVandenBroeckNJP2010, SagawaUedaPRL2010b, TakaraHasegawaDriebePLA2010}. 
To distinguish this case from the regime of validity of Eq.~(\ref{eq EP 5 micro}), we denote this inequality by 
\begin{equation}\label{eq EP 5 alternative}
 \tilde\Sigma_d(t) \equiv \Delta S_\text{vN}[\rho_S(t)] - \frac{Q(t)}{T_0} \ge 0.
\end{equation}
Importantly, for a bath far from equilibrium it has not been possible to link $Q(t)/T_0$ to an entropy change. 
Strictly speaking, Eq.~(\ref{eq EP 5 alternative}) therefore coincides with the second law only if the bath is weakly 
perturbed, whereas Eq.~(\ref{eq EP 4 micro}) is consistent with the second law for a larger class of 
transformations not restricted to the isothermal case. Furthermore, it was recently 
found~\cite{StrasbergDiazRieraCampenyPRE2021} that $\tilde\Sigma_d(t)$ is an upper bound on the entropy production since 
\begin{equation}
 \tilde\Sigma_d(t) - \Sigma_c(t) = D[\pi_B(\beta_t^*)\|\pi_B(\beta_0)] \ge 0,
\end{equation}
which has consequences for the efficiency of heat engines in contact with finite 
baths~\cite{StrasbergDiazRieraCampenyPRE2021}. 

\section{Further extensions}
\label{sec further extensions}

We here extend the previous framework to cover a larger class of initial states 
(Sec.~\ref{sec sub generalized states}), multiple baths (Sec.~\ref{sec sub multiple baths}) and particle transport 
(Sec.~\ref{sec sub particles}). 

\subsection{Generalized initial states}
\label{sec sub generalized states}

As promised above, the second law can be shown to strictly hold for a much larger class of initial states than 
those described by Eq.~(\ref{eq initial state OQS}). In fact, by Lemma~\ref{lemma 2nd law} we know that 
Eq.~(\ref{eq EP 2 micro}) holds for all initial states satisfying 
$S_\text{obs}^{S_0,E_B}[\rho_{SB}(0)] = S_\text{vN}[\rho(0)]$. By Lemma~\ref{lemma S obs equal S vN} and by choosing 
the coarse-graining from the previous section, these states are given by 
\begin{equation}\label{eq initial state general}
 \rho_{SB}(0) = \sum_{s_0,E_B} p_{s_0,E_B}(0)|s_0\rl s_0|\otimes\omega_B(E_B),
\end{equation}
with arbitrary probabilities $p_{s_0,E_B}(0)$. This generalizes the previous initial state~(\ref{eq initial state OQS}) 
in two ways. First, the bath need not be described by a Gibbs state---a microcanonical 
state or any convex combination thereof can also be considered. Second, the initial state does not need to be 
decorrelated. It can have arbitrary classical correlations with respect to the chosen coarse-graining. 

In view of what we said at the end of Sec.~\ref{sec sub Clausius}, it is also possible to imagine coarse-grainings 
different from the one chosen in Sec.~\ref{sec sub relevant coarse graining}. In particular, by going beyond a 
coarse-graining with a system-bath tensor product structure as considered here, quantum correlations could be 
included in the description. 

Finally, we explicitly decompose the entropy production~(\ref{eq EP 2 micro}) for an initial state of the 
form~(\ref{eq initial state general}) into all its contributions: 
\begin{align}
 \Sigma_a(t) =&~\Delta S_\text{obs}^{S_t}[\rho_S(t)] - \int \frac{\dbar Q(s)}{T_s^*} \label{eq EP OQS general} \\
 &+ S_\text{obs}^{E_B}[\rho_B(t)] - \C S_B(\beta_t^*) - S_\text{obs}^{E_B}[\rho_B(0)] + \C S_B(\beta_0^*) \nonumber \\
 &+ I_\text{obs}^{S_0,E_B}[\rho_{SB}(0)] - I_\text{obs}^{S_t,E_B}[\rho_{SB}(t)] \nonumber
\end{align}
The first line describes the Clausius contribution to the entropy production, obtained by neglecting system-bath 
correlations and by assuming the bath to be well described by its effective temperature only. The second line takes 
into account nonequilibrium features of the bath state in comparison with a fictitious Gibbs ensemble at the same 
energy. The third line quantifies the influence of system-bath correlations on the second law. 

To estimate the influence of each of these terms, we consider a small system, which is coupled to a large bath 
and subject to a, say, periodic driving protocol with period $\tau$. Furthermore, we consider times $t = n\tau$ 
with $n$ large. In this case, the system reaches a periodic steady state and constantly dissipates energy into the 
bath. We therefore expect that the entropy production scales with time such that $\Sigma_a(t) \sim t$. Our 
\emph{conjecture} is that the lines in Eq.~(\ref{eq EP OQS general}) have been ordered in decreasing relevance: 
\begin{align}
 &\Delta S_\text{obs}^{S_t}[\rho_S(t)] - \int \frac{\dbar Q(s)}{T_s^*} \\
 &\gg \left|S_\text{obs}^{E_B}[\rho_B(t)] - \C S_B(\beta_t^*) - S_\text{obs}^{E_B}[\rho_B(0)] + \C S_B(\beta_0^*)\right| \nonumber \\
 &\gg \left|I_\text{obs}^{S_0,E_B}[\rho_{SB}(0)] - I_\text{obs}^{S_t,E_B}[\rho_{SB}(t)]\right| \nonumber
\end{align}

We justify this conjecture as follows. First, if the system reaches a periodic steady state maintained by a constant 
uptake of mechanical work, the total heat flux $Q(t) \sim t$ has to scale proportional to $t$ by the first law. Thus, 
although $\Delta S_\text{obs}^{S_t}[\rho_S(t)]$ becomes negligible as it is bounded by $\ln(\dim\C H_S)$, the first 
line in Eq.~(\ref{eq EP OQS general}) is expected to scale as $t$. Furthermore, the last line can not scale with $t$ 
and must reach a constant, which is at most $2\ln(\dim\C H_S)$. Therefore, it is negligible for long times. The really 
challenging question concerns the second line. We can not exclude that this contribution scales with $t$, albeit we 
believe that its rate of growth should be in most cases \emph{sublinear} (e.g., $\sqrt{t}$). This believe is motivated 
by the fact that the microscopic dynamics of a typical heat bath are often very complex, characterized by (close to) 
chaotic behaviour, such that it becomes hard to distinguish its true state from an idealized Gibbs ensemble. This idea 
is indeed supported by research on equilibration and thermalization in isolated many-body 
systems~\cite{GemmerMichelMahlerBook2004, DAlessioEtAlAP2016, GogolinEisertRPP2016, GooldEtAlJPA2016, DeutschRPP2018, 
MoriEtAlJPB2018}. In any case, while the behaviour of the first and third line in Eq.~(\ref{eq EP OQS general}) 
appears universal, the behaviour of the second line will be model-dependent. 

\subsection{Multiple baths}
\label{sec sub multiple baths}

In many relevant situations, in particular to study transport process, the open system is coupled to multiple baths, 
labeled by $\nu\in\{1,\dots,n\}$, see Fig.~\ref{fig setup multiple baths} for a sketch. The system-bath 
Hamiltonian~(\ref{eq system bath Hamiltonian}) is then generalized to 
\begin{equation}
 H_{SB}(\lambda_t) = H_S(\lambda_t) + \sum_\nu \left[V_{SB}^{(\nu)} + H_B^{(\nu)}\right].
\end{equation}
We denote the global system-bath state at time $t$ by $\rho_{SB}(t)$ and the marginal state of 
bath $\nu$ by $\rho_\nu(t)$. In the following, we show that our framework can be extended to this situation in a 
straightforward way. 

\begin{figure}[b]
 \centering\includegraphics[width=0.44\textwidth,clip=true]{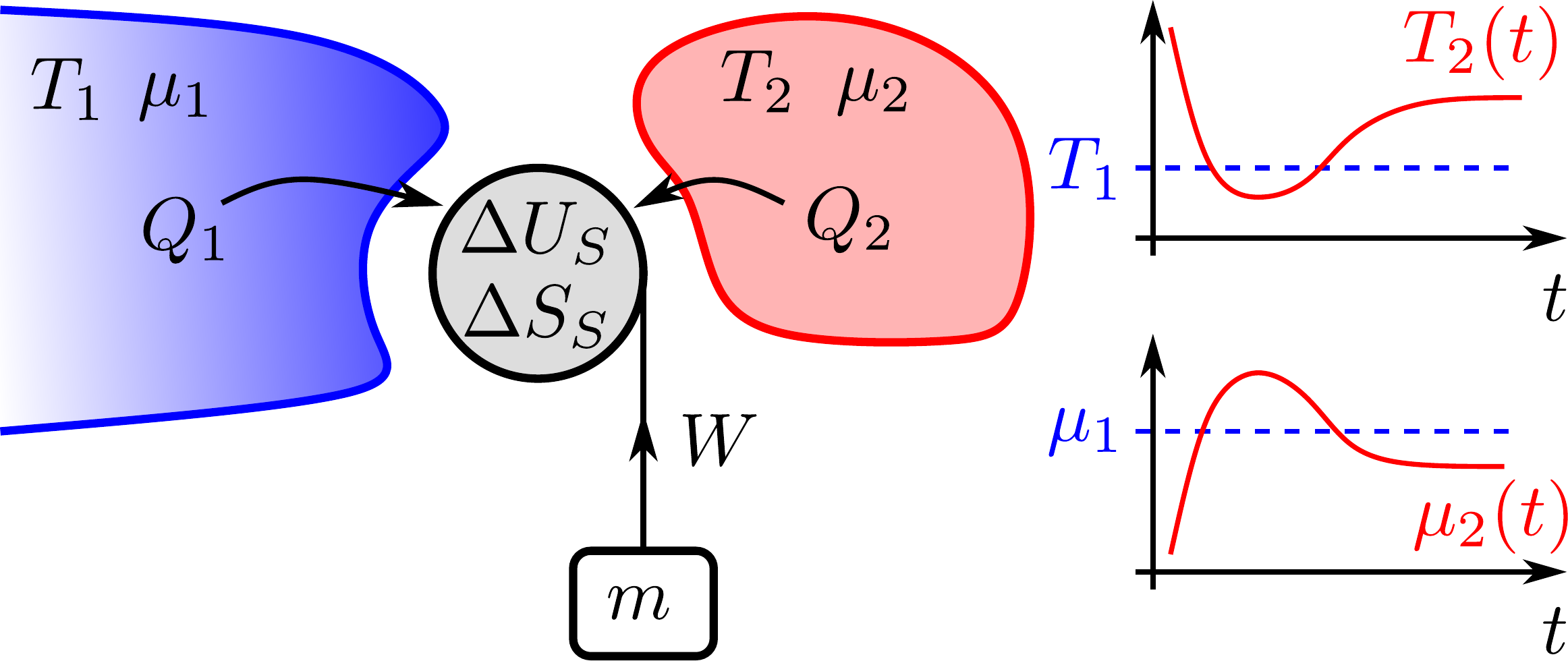}
 \label{fig setup multiple baths} 
 \caption{Sketch of a system in contact with an infinite bath at temperature $T_1$ and chemical potential $\mu_1$, 
 a finite bath at $T_2$ and $\mu_2$ and a work reservoir (sketched by a weight $m$). The system experiences a change 
 in internal energy $\Delta U_S$ and system entropy $\Delta S_S$ due to heat flows $Q_1$ and $Q_2$ from the bath and 
 the mechanical work $W$ supplied to it. Since the second bath is finite, its temperature and chemical potential 
 change with time $t$ whereas $T_1$ and $\mu_1$ remain unchanged (sketched on the right). }
\end{figure}

First, as our relevant coarse-graining we choose $\{|s_t\rl s_t|\otimes\Pi_{E_1}\otimes\dots\otimes\Pi_{E_n}\}$, 
where $\Pi_{E_\nu}$ corresponds to a coarse-grained measurement of the energy of bath $\nu$. For notational simplicity, 
we write $\bb E \equiv (E_1,\dots,E_n)$. Then, the observational entropy is generalized to 
\begin{equation}\label{eq obs ent multiple bath}
 S_\text{obs}^{S_t,\bb E}[\rho_{SB}(t)] = -\sum_{s_t,\bb E} p_{s_t,\bb E}(t)\ln\frac{p_{s_t,\bb E}(t)}{V_{\bb E}}
\end{equation}
with $V_{\bb E} \equiv \prod_\nu \mbox{tr}_{B_\nu}\{\Pi_{E_\nu}\}$. The initial state of our setup is described 
by a generalization of the initial state~(\ref{eq initial state OQS}), 
\begin{equation}\label{eq initial state multiple baths}
 \rho_{S\bb B}(0) = \rho_S(0)\otimes \pi_1(\beta_1)\otimes\dots\otimes\pi_n(\beta_n),
\end{equation}
assuming each bath to be prepared in a Gibbs ensemble at inverse temperature $\beta_\nu$. Clearly, from 
Sec.~\ref{sec sub generalized states} we know that a larger class of initial states is admissible. 

From these considerations, a non-negative change in thermodynamic entropy quantified by 
Eq.~(\ref{eq obs ent multiple bath}) follows: 
\begin{equation}
 \Sigma_a(t) = S_\text{obs}^{S_t,\bb E}[\rho_{SB}(t)] - S_\text{obs}^{S_0,\bb E}[\rho_{SB}(0)] \ge 0.
\end{equation}
Since the initial state is decorrelated, we also confirm 
\begin{equation}
 \Sigma_b(t) = \Delta S_\text{obs}^{S_t}[\rho_S(t)] + \sum_\nu \Delta S_\text{obs}^{E_\nu}[\rho_\nu(t)] \ge 0.
\end{equation}
Importantly, the difference 
\begin{equation}
 \begin{split}
  \Sigma_b(t) - \Sigma_a(t) =&~ S_\text{obs}^{S_t}[\rho_S(t)] + \sum_\nu S_\text{obs}^{E_\nu}[\rho_\nu(t)] \\
  &- S_\text{obs}^{S_t,\bb E}[\rho_{SB}(t)]
 \end{split}
\end{equation}
is now given by the non-negative \emph{total information}, which---even if $\dim\C H_S$ is small---can be large as it 
also quantifies the correlations between the different baths. 

Next, we use our definition of temperature, Eq.~(\ref{eq def noneq temp}), for each bath separately. Then, if 
$(T_t^*)_\nu$ describes the effective nonequilibrium temperature of bath $\nu$ at time $t$, manipulations identical to 
those of Sec.~\ref{sec sub Clausius} yield 
\begin{equation}\label{eq EP 4 micro multi}
 \Sigma_c(t) = S_\text{obs}^{S_t}[\rho_S(t)] - \sum_\nu \int \frac{\dbar Q_\nu(s)}{(T_s^*)_\nu} \ge 0.
\end{equation}
Here, $\dbar Q_\nu(t) = -dU_{B_\nu}(t)$ is minus the infinitesimal change in energy of bath $\nu$. This identification 
of heat implies the first law 
\begin{equation}
 \Delta U_S(t) = \sum_\nu Q_\nu(t) + W(t),
\end{equation}
where $U_S(t) \equiv \mbox{tr}_{SB}\{[H_S(\lambda_t) + \sum_\nu V^{(\nu)}_{SB}]\rho_{SB}(t)\}$ generalizes 
Eq.~(\ref{eq def internal energy open}) and $W(t)$ is still given by Eq.~(\ref{eq def work OQS}). 

Finally, we consider the case of very large baths or, alternatively, times $t$ that a short enough such that the baths 
are only weakly perturbed. Then, $(T_t^*)_\nu \approx T_\nu$ and Eq.~(\ref{eq EP 4 micro multi}) reduces to 
\begin{equation}\label{eq EP 5 micro multi}
 \Sigma_d(t) = S_\text{obs}^{S_t}[\rho_S(t)] - \sum_\nu \frac{Q_\nu(t)}{T_\nu} \ge 0.
\end{equation}
Of course, one could also imagine situations where the baths have different sizes and need to be treated accordingly. 
Moreover, in the long run and if the bath size is finite, one expects all baths to equilibrate to the same 
temperature. This behaviour is captured by Eq.~(\ref{eq EP 4 micro multi}), but not by Eq.~(\ref{eq EP 5 micro multi}). 

\subsection{Particle exchanges}
\label{sec sub particles}

Energy is not the only quantity, which gets exchanged between different subsystems. Also particles are exchanged and 
the most relevant particle species for current nanotechnological applications are probably electrons. To avoid 
notational clutter, the exposition below is adapted to the case of a \emph{single} particle species only. 

We start with equilibrium considerations for an isolated system with Hamiltonian $H$ and particle number operator 
$\hat N$. We use a `hat' for the particle number operator to distinguish it from its expectation value 
$N \equiv \mbox{tr}\{\hat N\rho\}$. At equilibrium, the theory is constructed by using the grand canonical ensemble 
\begin{equation}
 \Xi(\beta,\mu) \equiv \frac{e^{-\beta(H-\mu\hat N)}}{\C Z(\beta,\mu)},
\end{equation}
where $\mu$ is the chemical potential and $\C Z(\beta,\mu) \equiv \mbox{tr}\{e^{-\beta(H-\mu\hat N)}\}$ the 
grand canonical partition function. An infinitesimal change in the equilibrium entropy 
$\C S(\beta,\mu) = S_\text{vN}[\Xi(\beta,\mu)]$ can be expressed as 
\begin{equation}\label{eq dS grand canonical}
 d\C S = \frac{1}{T}d\C U - \frac{\mu}{T}dN.
\end{equation}

In unison with our definition of an effective nonequilibrium temperature, we can also introduce an effective 
chemical potential for any state $\rho$ by demanding that its particle number expectation value matches the one 
of the grand canonical ensemble. Thus, the following two equations determine $\beta^*$ and $\mu^*$: 
\begin{align}
 U = \mbox{tr}\{H\rho\}         &\equiv \mbox{tr}\{H\Xi(\beta^*,\mu^*)\}, \\
 N = \mbox{tr}\{\hat N\rho\}    &\equiv \mbox{tr}\{\hat N\Xi(\beta^*,\mu^*)\}.
\end{align}
To connect the equilibrium entropies of two states with $(\beta^*_t,\mu^*_t)$ and $(\beta^*_0,\mu^*_0)$, we find 
from Eq.~(\ref{eq dS grand canonical}) and in accordance with Eq.~(\ref{eq entropy Clausius term}) that 
\begin{equation}\label{eq Delta S grand canonical}
 \C S(\beta_t^*,\mu_t^*) - \C S(\beta_0^*,\mu_0^*) = \int \frac{dU(s) - \mu_s^* dN(s)}{T_s^*}.
\end{equation}

Finally, we define observational entropy with respect to a coarse-graining of energy and particles. 
Since $[H,\hat N] = 0$, we can jointly measure both quantities. Then, 
\begin{equation}
 S_\text{obs}^{E,N}(\rho) = \sum_{E,N} p_{E,N}(-\ln p_{E,N} + \ln V_{E,N}).
\end{equation}
Each quantity is defined by analogy with the previous case: $p_{E,N} = \mbox{tr}\{\Pi_E\Pi_N\rho\}$ and 
$V_{E,N} = \mbox{tr}\{\Pi_E\Pi_N\}$. Note that both, $\Pi_E$ and $\Pi_N$, describe in general again a measurement with 
a finite resolution or uncertainty. As before, however, we demand that the uncertainty is small enough 
to be consistent at equilibrium such that $S_\text{obs}^{E,N}[\Xi(\beta,\mu)] \approx \C S(\beta,\mu)$. 

After these preliminary consideration, we can now return to the case of a system coupled to multiple baths, exchanging 
energy and particles with them, see Fig.~\ref{fig setup multiple baths}. 
Equations~(\ref{eq dS grand canonical}) and~(\ref{eq Delta S grand canonical}) 
suggest to define the infinitesimal heat flux from bath $\nu$ at time $t$ as 
\begin{equation}
 \dbar Q_\nu(t) \equiv -[dU_\nu(t) - \mu_\nu^* dN_\nu(t)],
\end{equation}
even if the state $\rho_\nu(t)$ of bath $\nu$ is out of equilibrium. From this definition it follows that the 
first law needs to be generalized to 
\begin{equation}
 \Delta U_S(t) = \sum_\nu Q_\nu(t) + W(t) + W_\text{chem}(t).
\end{equation}
Here, a new contribution appears known as \emph{chemical work}. It is defined as 
$W_\text{chem}(t) = \int \dbar W_\text{chem}(s)$ with $\dbar W_\text{chem}(s) = \sum_\nu (\mu_s^*)_\nu dN_\nu$, 
where $(\mu_s^*)_\nu$ denotes the chemical potential of bath $\nu$ at time $s$. This form of work is 
associated with particle exchanges and quantifies the ability of, e.g., electrons to charge a battery, whose 
energy can be converted back into mechanical work. 

Finally, by choosing the coarse-graining 
\begin{equation}
 \{|s_t\rl s_t|\otimes \Pi_{E_1}\Pi_{N_1} \otimes\dots\otimes \Pi_{E_n}\Pi_{N_n}\},
\end{equation}
observational entropy becomes 
\begin{equation}
 S_\text{obs}^{S_t,\bb E,\bb N}[\rho_{SB}(t)] 
 = -\sum_{s_t,\bb E,\bb N} p_{s_t,\bb E,\bb N}(t) \ln\frac{p_{s_t,\bb E,\bb N}(t)}{V_{s_t,\bb E,\bb N}(t)}.
\end{equation}
The definition of each term should be obvious as it follows by analogy with the previous cases. Furthermore, 
by restricting our considerations to the initial state 
\begin{equation}
 \rho_{SB}(0) = \rho_S(0)\otimes\Xi(\beta_1,\mu_1)\otimes\dots\otimes\Xi(\beta_n,\mu_n),
\end{equation}
the following hierarchy of second laws also follows by analogy: 
\begin{align}
 0  &\le \Sigma_a(t) \equiv \Delta S_\text{obs}^{S_t,\bb E,\bb N}[\rho_{SB}(t)] \\
    &\le \Sigma_b(t) \equiv 
    \Delta S_\text{obs}^{S_t}[\rho_S(t)] + \sum_\nu \Delta S_\text{obs}^{E_\nu,N_\nu}[\rho_\nu(t)] \\
    &\le \Sigma_c(t) \equiv \Delta S_\text{obs}^{S_t}[\rho_S(t)] - \int \frac{\dbar Q_\nu(s)}{(T_s^*)_\nu}. 
    \label{eq EP 4 micro N}
\end{align}
Again, it is possible to quantify the difference between $\Sigma_a(t)$ and $\Sigma_b(t)$ by the total information 
and between $\Sigma_b(t)$ and $\Sigma_c(t)$ by nonequilibrium features in the bath distribution. Finally, by 
considering the limit of a weakly perturbed bath described by $(T_t^*)_\nu \approx T_\nu$ and 
$(\mu_t^*)_\nu \approx \mu_\nu$, we obtain from Eq.~(\ref{eq EP 4 micro N}) 
\begin{equation}\label{eq EP 5 micro N}
 \Sigma_c(t) \approx \Sigma_d(t) = \Delta S_\text{obs}^{S_t}[\rho_S(t)] - \frac{Q_\nu(t)}{T_\nu} \ge 0
\end{equation}
with $Q_\nu(t) = -[\Delta U_\nu(t) - \mu_\nu\Delta N_\nu(t)]$. Equation~(\ref{eq EP 5 micro N}) quantifies the 
entropy production for transport processes, where the baths are kept at a \emph{fixed} temperature and chemical 
potential. The scope of Eq.~(\ref{eq EP 4 micro N}) is wider and captures dynamical features in the bath, already 
observed in experiments~\cite{BrantutEtAlScience2012, BrantutEtAlScience2013}, in a self-contained way. 

\section{Fluctuation theorems}
\label{sec fluctuation theorems}

Fluctuation theorems present important refinements on our view on the second law. They are exact relations, which 
constrain the fluctuations in thermodynamics quantities such that, among other consequences, the second law can be 
formulated as an \emph{equality}. Fluctuations theorems play an important role in classical nonequilibrium 
statistical mechanics~\cite{EvansSearlesAdvPhy2002, PitaevskiiPU2011, JarzynskiAnnuRevCondMat2011}, stochastic 
thermodynamics~\cite{SeifertRPP2012, VandenBroeckEspositoPhysA2015} and quantum thermodynamics based on the so-called 
`two-point measurement scheme'~\cite{EspositoHarbolaMukamelRMP2009}. 

The goal of this section is to show that the entropy production as defined by the change of observational entropy also 
satisfies fluctuations theorems. We do so in an abstract way as in Sec.~\ref{sec sub obs ent}, assuming the entire 
system is isolated and evolves according to Eq.~(\ref{eq time evo isolated}). We believe that the derivation below 
captures the essence of fluctuation theorems from a \emph{technical} point of view. Particular applications can be 
then worked out by following the lines of Secs.~\ref{sec entropy second law isolated}, \ref{sec first second law open} 
and~\ref{sec further extensions}, which we will not do here. 

To approach the problem, we first define fluctuations of observational entropy. From definition~(\ref{eq def obs ent}) 
we see that observational entropy can be written as an average of 
\begin{equation}\label{eq obs ent fluct}
 s_\text{obs}(x,t) \equiv -\ln p_x(t) + \ln V_x,
\end{equation}
where the average is carried out with respect to the probabilities $p_x(t)$: 
\begin{equation}
 S_\text{obs}^X[\rho(t)] = \sum_x p_x(t) s_\text{obs}(x,t). 
\end{equation}
Thus, $s_\text{obs}(x,t)$ is a random variable, whose construction requires knowledge of the probabilities $p_x(t)$. 

Next, we look at fluctuations in the \emph{change} of $s_\text{obs}(x,t)$. To this end, we use the two-point 
measurement scheme, first used in Refs.~\cite{PiechocinskaPRA2000, KurchanArXiv2000, TasakiArXiv2000}. 
Imagine that we perform initially a measurement of $X_0$, giving rise to outcome $x_0$, and finally a measurement of 
$X_t$ with outcome $x_t$. The fluctuations of the random variable~(\ref{eq obs ent fluct}) in this process are 
\begin{equation}\label{eq stoch obs ent change}
 \begin{split}
  \Delta s_\text{obs}(x_t,t;x_0,0) &\equiv s_\text{obs}(x_t,t) - s_\text{obs}(x_0,0) \\
  &= \ln\frac{p_{x_0}(0)V_{x_t}}{V_{x_0}p_{x_t}(t)}. 
 \end{split}
\end{equation}
Moreover, the probability to observe outcomes $(x_t,x_0)$ is 
\begin{equation}\label{eq forward two time prob}
 p_{x_t,x_0} = \mbox{tr}\{\Pi_{x_t}U(t,0)\Pi_{x_0}\rho(0)\Pi_{x_0}U^\dagger(t,0)\}.
\end{equation}
Finally, let us denote by $\lr{\dots} = \sum_{x_t,x_0} \dots p_{x_t,x_0}$ an average over this process. 

Then, if the condition $S_\text{obs}^{X_0}[\rho(0)] = S_\text{vN}[\rho(0)]$ is satisifed 
(which we also assumed to derive our second laws, cf.~Lemma~\ref{lemma 2nd law}), we find the following 
\emph{integral fluctuation theorem}: 
\begin{equation}\label{eq fluctuation theorem}
 \lr{e^{-\Delta s_\text{obs}(x_t,t;x_0,0)}} = 1,
\end{equation}
where here and in the following we tacitly assume $p_{x_0}\neq0$ for all $x_0$ to avoid `dividing by zero,' 
which is related to the phenomenon of absolute irreversibility~\cite{MurashitaFunoUedaPRE2014}. 

The proof goes as follows. From Eq.~(\ref{eq stoch obs ent change}) and 
the assumption $\rho(0) = \sum_{x_0} p_{x_0}\Pi_{x_0}/V_{x_0}$, which implies 
$\Pi_{x_0}\rho(0)\Pi_{x_0} = p_{x_0}\Pi_{x_0}/V_{x_0}$, we get the chain of equalities: 
\begin{align}
 \lr{e^{-\Delta s_\text{obs}(x_t,t;x_0,0)}} 
 &= \sum_{x_t,x_0} \mbox{tr}\{\Pi_{x_t}U(t,0)\Pi_{x_0}U^\dagger(t,0)\}\frac{p_{x_t}}{V_{x_t}} \nonumber \\
 &= \sum_{x_t} \mbox{tr}\{\Pi_{x_t}U(t,0)U^\dagger(t,0)\}\frac{p_{x_t}}{V_{x_t}} \nonumber \\
 &= \sum_{x_t} \mbox{tr}\{\Pi_{x_t}\}\frac{p_{x_t}}{V_{x_t}} = 1.
\end{align}
For the last steps we used $\sum_{x_0}\Pi_{x_0} = 1$, $U(t,0)U^\dagger(t,0) = 1$, $\mbox{tr}\{\Pi_{x_t}\} = V_{x_t}$, 
and $\sum_{x_t} p_{x_t} = 1$. 

By using the inequality $e^y \ge 1+y$ for $y\in\mathbb{R}$, we confirm that the integral fluctuation 
theorem~(\ref{eq fluctuation theorem}) implies the formal second law~(\ref{eq 2nd law formal}). An even more general 
class of integral fluctuation theorems was derived in Ref.~\cite{SchmidtGemmerZNA2020}. 

Finally, there is also a \emph{detailed fluctuation theorem}, which makes the connection with time-reversal 
symmetry (see Appendix~\ref{sec time reversal}) particularly transparent and implies the integral fluctuation 
theorem. To derive it, we start with the probability $P(\Delta s)$ to observe a change in 
observational entropy $\Delta s$ in the \emph{forward process}: 
\begin{equation}\label{eq forward probability}
 P_\text{fw}(\Delta s) = 
 \sum_{x_t,x_0} \delta[\Delta s - \Delta s_\text{obs}^\text{fw}(x_t,t;x_0,0)] p^\text{fw}_{x_t,x_0},
\end{equation}
where $\delta(\cdot)$ denotes the Dirac-delta function. Here, we have been particularly cautious and indicated 
by `fw' quantities associated to the forward process.

Next, we introduce the \emph{`time-reversed' process}. To this end, we use time-reversal symmetry 
(see Appendix~\ref{sec time reversal}) and denote the time-reversed projectors by 
$\Pi_x^\Theta \equiv \Theta \Pi_x \Theta^{-1}$ and by $U_\Theta(t,0)$ the unitary time evolution operator 
associated to the time-reversed dynamics. The time-reversed process is then defined by starting with a 
measurement of $\Pi_{x_t}^\Theta$, followed by an evolution according to $U_\Theta(t,0)$, and ending with a 
measurement of $\Pi_{x_0}^\Theta$. The probability to observe the sequence of measurement results $(x_0,x_t)$ in the 
time-reversed process is 
\begin{equation}
 p_{x_0,x_t}^\text{tr} \equiv 
 \mbox{tr}\{\Pi^\Theta_{x_0}U_\Theta(t,0)\Pi^\Theta_{x_t}\rho_\text{tr}(t)\Pi^\Theta_{x_t}U_\Theta^\dagger(t,0)\},
\end{equation}
where $\rho_\text{tr}(t)$ is the initial state in the time-reversed process. Note that we count time `backwards' in 
the time-reversed process, starting at $t$ and ending at $0$, which is convenient from a notational perspective. We 
emphasize, however, that in any experimental realization of that process time runs `as always' forward. 

As done multiple times before, we assume again that the initial states in the forward and time-reversed process obey 
$S_\text{obs}^{X_0}[\rho_\text{fw}(0)] = S_\text{vN}[\rho_\text{fw}(0)]$ and 
$S_\text{obs}^{X_t}[\rho_\text{tr}(t)] = S_\text{vN}[\rho_\text{tr}(t)]$. This implies that (see 
Lemma~\ref{lemma S obs equal S vN}) $\rho_\text{fw}(0) = \sum_{x_0} p^\text{fw}_{x_0} \Pi_{x_0}/V_{x_0}$ 
and $\rho_\text{tr}(t) = \sum_{x_0} p^\text{tr}_{x_t} \Pi_{x_t}^\Theta/V_{x_t}$ for arbitrary probabilities 
$p^\text{fw}_{x_0}$ and $p_{x_t}^\text{tr}$. We now make the important choice that 
\begin{equation}\label{eq special choice}
 p^\text{tr}_{x_t} = p^\text{fw}_{x_t},
\end{equation}
i.e., the initial probabilities in the time-reversed process coincide with the final measurement statistics of the 
forward process. Note that this does \emph{not} imply that $\rho_\text{tr}(t)$ is the time-reversed final state of 
the forward process, i.e., $\rho_\text{tr}(t) \neq \Theta\rho_\text{fw}(t)\Theta^{-1}$ with 
$\rho_\text{fw}(t) = U(t,0)\rho_\text{fw}(0)U^\dagger(t,0)$. Taken together, these 
assumptions and our special choice imply the central relation 
\begin{equation}\label{eq central FT}
 p^\text{fw}_{x_t,x_0} = \exp[\Delta s_\text{obs}^\text{fw}(x_t,t;x_0,0)] p_{x_0,x_t}^\text{tr}.
\end{equation}
This result follows from the relation 
\begin{equation}
 \begin{split}
  \mbox{tr}\{&\Pi_{x_0}U^\dagger(t,0)\Pi_{x_t}U(t,0)\} \\
  &= \mbox{tr}\{\Pi^\Theta_{x_0}U_\Theta(t,0)\Pi^\Theta_{x_t}U_\Theta^\dagger(t,0)\}
 \end{split}
\end{equation}
which is a consequence of Eqs.~(\ref{eq trace conjugation}) and~(\ref{eq time reversal unitary}) and 
$\mbox{tr}\{\Pi_{x_0}U^\dagger(t,0)\Pi_{x_t}U(t,0)\} \in \mathbb{R}$. 

Now, we return to Eq.~(\ref{eq forward probability}). From Eq.~(\ref{eq central FT}) we immediately obtain 
\begin{equation}
 P_\text{fw}(\Delta s) = e^{\Delta s} 
 \sum_{x_t,x_0} \delta[\Delta s - \Delta s_\text{obs}^\text{fw}(x_t,t;x_0,0)] p_{x_0,x_t}^\text{tr},
\end{equation}
where we used the Dirac-delta function to pull the factor $e^{\Delta s}$ out of the summation. Next, we note 
that $\Delta s_\text{obs}^\text{fw}(x_t,t;x_0,0) = - \Delta s_\text{obs}^\text{fw}(x_0,0;x_t,t)$, which implies 
\begin{align}\label{eq pre DFT}
 P_\text{fw}(\Delta s) &= e^{\Delta s} 
 \sum_{x_t,x_0} \delta[\Delta s + \Delta s_\text{obs}^\text{fw}(x_0,0;x_t,t)] p_{x_0,x_t}^\text{tr} \nonumber \\
 &\equiv e^{\Delta s} Q_\text{tr}(-\Delta s).
\end{align}
Here, $Q_\text{tr}(\Delta s)$ is the probability that the quantity $\Delta s_\text{obs}^\text{fw}(x_0,0;x_t,t)$ 
takes on the value $\Delta s$ in the time-reversed process with respect to the choice~(\ref{eq special choice}). 
This choice also reveals that  
\begin{align}
 \Delta s_\text{obs}^\text{tr}(x_0,0;x_t,t) 
 &= \ln\frac{p_{x_0}^\text{tr}(0) V_{x_t}}{V_{x_0} p_{x_t}^\text{tr}(t)} \\
 &= \ln\frac{p_{x_0}^\text{fw}(0)}{p_{x_0}^\text{tr}(0)} + \Delta s_\text{obs}^\text{fw}(x_0,0;x_t,t). \nonumber
\end{align}
In general, $p_{x_0}^\text{fw}(0) \neq p_{x_0}^\text{tr}(0)$ and then $Q_\text{tr}(\Delta s)$ can \emph{not} be 
interpreted as the distribution of the stochastic entropy production associated to the time-reversed process. Instead, 
$Q_\text{tr}(\Delta s)$ quantifies the distribution of entropy production associated to the forward dynamics, 
according to a fixed choice of $p_{x_0}^\text{fw}(0)$, but which is observed during the time-reversed process. 
Nevertheless, Eq.~(\ref{eq pre DFT}) implies the integral fluctuation theorem~(\ref{eq fluctuation theorem}) and, 
therefore, provides a stronger constraint on $P_\text{fw}(\Delta s)$ than Eq~(\ref{eq fluctuation theorem}). 

Finding the conditions for which $p_{x_0}^\text{fw}(0) = p_{x_0}^\text{tr}(0)$ holds is not simple in our 
general setting and involves additional assumptions (e.g., steady state regime or relaxation to equilibrium 
after the driving)~\cite{SeifertRPP2012, VandenBroeckEspositoPhysA2015}. However, if 
$p_{x_0}^\text{fw}(0) = p_{x_0}^\text{tr}(0)$, then $Q_\text{tr}(\Delta s) \equiv P_\text{tr}(\Delta s)$ is the 
distribution of entropy production during the time-reversed process and we find the detailed fluctuation theorem 
\begin{equation}
 \frac{P_\text{fw}(\Delta s)}{P_\text{tr}(-\Delta s)} = e^{\Delta s}.
\end{equation}

\section{Concluding reflections}
\label{sec final}

This tutorial was devoted to the understanding, derivation and quantification of the laws of thermodynamics in open 
and isolated quantum systems based on microscopic definitions for internal energy, heat, work, 
(thermodynamic) entropy, temperature and chemical potentials. Summarizing the situation for open systems without 
particle exchanges, the first law reads $\Delta U_S(t) = Q(t) + W(t)$. Moreover, using 
Eq.~(\ref{eq obs ent OQS 1 bath}) as our entropy definition and Eq.~(\ref{eq def noneq temp}) as our definitions for 
the (nonequilibrium) temperature of the bath, we found that the second law for a system initially decorrelated from 
a thermal bath can be summarized by the following hierarchy of inequalities, where each member reflects the degree of 
control or information taken into account in an experiment: 
\begin{widetext}
 \begin{align}
  0 &\le \Delta S_\text{obs}^{S_t,E_B}[\rho_{SB}(t)] 
    &\text{most general version of the second law} \\
    &\le \Delta S_\text{obs}^{S_t}[\rho_S(t)] + \Delta S_\text{obs}^{E_B}[\rho_B(t)] 
    &\text{disregard $SB$ correlations $I_\text{obs}^{S_t,E_B}[\rho_{SB}(t)]$} \\
    &\le \Delta S_\text{obs}^{S_t}[\rho_S(t)] - \int \frac{\dbar Q(s)}{T_s^*}
    &\text{disregard noneq.~bath distribution $S_\text{vN}[\pi_B(\beta_t^*)] - S_\text{obs}^{E_B}[\rho_B(t)]$} \\
    &\le \Delta S_\text{obs}^{S_t}[\rho_S(t)] - \frac{Q(t)}{T_0} 
    &\text{disregard finite-size effects $D[\pi_B(\beta_t^*)\|\pi_B(\beta_0)]$}
 \end{align}
\end{widetext}
This exactly matches the hierarchy of phenomenological second laws~(\ref{eq EP basic basic}), (\ref{eq EP basic}), 
(\ref{eq EP traditional}) and~(\ref{eq EP traditional OQS}) if one identifies observational entropy with 
thermodynamic entropy, as also done in Refs.~\cite{VonNeumann1929, VonNeumannEPJH2010, PercivalJMP1961, 
LatoraBarangerPRL1999, NauenbergAJP2004, LeePRE2018, SafranekDeutschAguirrePRA2019a, SafranekDeutschAguirrePRA2019b, 
StrasbergArXiv2019b, SafranekAguirreDeutschPRE2020, FaiezEtAlPRA2020, SchindlerSafranekAguirrePRA2020, 
RieraCampenySanperaStrasbergPRXQ2021, StrasbergDiazRieraCampenyPRE2021}. Thus, by starting with a microscopic 
definition of thermodynamic entropy, a conceptually clear and consistent approach emerges, which covers 
diverse applications such as multiple baths, initially correlated or non-Gibbsian bath states, small 
baths with changing temperatures, \emph{etc}. 

In fact, our growing nanotechnological abilities also enhance our abilities to control and measure a bath. Finite size 
effects and changing thermodynamic parameters are already reality in experiments~\cite{TrotzkyEtAlNP2012, 
BrantutEtAlScience2012, GringEtAlScience2012, BrantutEtAlScience2013, ClosEtAlPRL2016, KaufmanEtAlScience2016, 
KrinnerEsslingerBrantutJP2017, BohlenEtAlPRL2020}. Furthermore, ultrasensitive thermometers were developed to track 
small changes in bath energies~\cite{GoveniusEtAlPRB2014, GasparinettiEtAlPRAppl2015, HalbertalEtAlNat2016, 
KarimiEtAlNC2020}. Observational entropy explicitly takes into account experimental (in)capabilities in its definition 
from the start. Furthermore, notice that the initial states chosen in this tutorial, e.g., in 
Eqs.~(\ref{eq initial state OQS}) or~(\ref{eq initial state general}) or see Lemma~\ref{lemma S obs equal S vN}, 
satisfy the requirement that an initial measurement of the chosen coarse-graining does \emph{not} disturb the state. 
Hence, the present approach is suitable to study a variety of thermodynamic processes at the nanoscale and can be 
readily applied to many experiments. 

Of course, not in every experiment will it be possible to measure the variables or observables that we have 
chosen above to derive the second law. Moreover, even if it is possible to measure those observables, the 
coarse-graining might not be `fine' enough and deviations of the actual initial state from our choice above, 
Eqs.~(\ref{eq initial state OQS}) or~(\ref{eq initial state general}), could become visible. In these cases, there is 
no guarantee that the change in observational entropy remains always non-negative and satisfies a second law. Still, 
we believe that 
one should not be afraid of such `violations' of the second law. Instead, one should view them as a \emph{welcome} 
feature as they reveal something \emph{unexpected} about the experimental setup. To quote Jaynes 
again~\cite{JaynesInBook1992}: ``recognizing this should increase rather than decrease our confidence 
in the future of the second law, because it means that if an experiment ever sees an apparent violation, then instead 
of issuing a sensational announcement, it will be more prudent to search for that unobserved degree of freedom.'' 

Despite experimental applications, we believe that also much fruitful theoretical work lies ahead of us. 
For instance, an observer's (in)capabilities to precisely know or control an experimental setup are also a central 
element of the resource theory of thermodynamics~\cite{GourMuellerEtAlPR2015, GooldEtAlJPA2016, NgWoodsBook2018, 
LostaglioRPP2019}. In there, additional constraints on the global unitary dynamics are imposed, for instance, by 
demanding that $[U(t,0),H_S+H_B] = 0$. In the present approach, we have put no constraints on the dynamics, but 
rather focused on constraints on the available knowledge. For instance, if the bath energy is only known up to a 
small uncertainty, it is experimentally impossible to determine whether $[U(t,0),H_S+H_B] = 0$ strictly holds. 
Combining both aspects could prove very fruitful to equip the present framework with more predictive power while 
making the resource theory approach more applicable in practice. In fact, the latter seems still in search of 
practical experimental applications~\cite{HalpernChap2017}.

It is also desirable to extend the present approach in other directions. For instance, the role of 
non-commuting coarse-grainings does not yet appear to be fully understood. Moreover, the tutorial focused on 
`two-time statistics' characterized by a non-disturbing initial and a final measurement. It remains unclear 
what is the effect of multiple sequential measurements~\cite{MilzModiPRXQ2021}, but see 
Ref.~\cite{GemmerSteinigewegPRE2014} for preliminary results. It is also interesting to extend the present framework 
to more generalized measurements characterized by positive operator-valued measures 
(`POVMs'). In fact, strict projective measurements are hard to realize in an experiment. 
More likely is that a measurement outcome $x'$ corresponds to applying a Gaussian weight of projectors $\Pi_x$ fixed 
around $x\approx x'$. Interestingly, for an arbitrary set of POVM elements $\{P_x\}$, which always satisfy 
$\sum_x P_x = 1$, the main definition~(\ref{eq def obs ent}) of observational entropy remains: the probability to 
observe outcome $x$ is given by $p_x = \mbox{tr}\{P_x\rho\}$ and the volume term becomes $V_x = \mbox{tr}\{P_x\}$. 
Therefore, it seems that the same qualitative picture should emerge, but this requires further research. 
It also seems that the current framework could inspire research in open quantum system theory (see, e.g, 
Ref.~\cite{RieraCampenySanperaStrasbergPRXQ2021}) and it has much potential to be fruitfully combined with methods 
reviewed in Refs.~\cite{GemmerMichelMahlerBook2004, DAlessioEtAlAP2016, GogolinEisertRPP2016, GooldEtAlJPA2016, 
DeutschRPP2018, MoriEtAlJPB2018} to study the equilibration and thermalization dynamics of isolated many-body systems. 

Finally, one can question whether---out of the plethora of possible candidates---our choice to use observational 
entropy as a microscopic definition for thermodynamic entropy was correct. We believe that this is not 
definitely answered by the present tutorial, but we also believe that it has added significant appeal to this 
definition. Therefore, it seems that the final answer to that question can not be too far from 
the present considerations. 

Thus, to conclude, observational entropy is a versatile concept, which provides a link between problems studied in 
the field of equilibration and thermalization in isolated quantum systems, quantum thermodynamics and open quantum 
systems theory. Therefore, it has the potential to provide an overarching framework for many problems studied in 
nonequilibrium statistical mechanics. 

\subsection*{Acknowledgements}

We are grateful to Kavan Modi, Juan Parrondo, Felix Pollock, Andreu Riera-Campeny, Dominik \u{S}afr\'anek, Anna 
Sanpera, and Joan Vaccaro for many interesting discussions. We also thank Miquel Saucedo Cuesta and Pablo 
Torron Perez for useful comments on the manuscript. PS further acknowledges various stimulating discussions with 
Massimiliano Esposito about the nature of entropy production over the years. The authors were partially supported by 
the Spanish Agencia Estatal de Investigaci\'on (project PID2019-107609GB-I00 and IJC2019-040883-I), the Spanish MINECO 
(FIS2016-80681-P, AEI/FEDER, UE), and the Generalitat de Catalunya (CIRIT 2017-SGR-1127). 
PS is financially supported by the DFG (project STR 1505/2-1). 


\bibliography{/home/philipp/Documents/references/books,/home/philipp/Documents/references/open_systems,/home/philipp/Documents/references/thermo,/home/philipp/Documents/references/info_thermo,/home/philipp/Documents/references/general_QM,/home/philipp/Documents/references/math_phys,/home/philipp/Documents/references/equilibration}

\appendix
\section{The time-reversal operator}
\label{sec time reversal}

In classical mechanics, it is clear from intuition that the trajectory of a particle is traversed in the opposite 
direction if one flips the momentum $p$ of the particle to $-p$. More precisely, if one follows the trajectory in phase 
space during a time window $[0,t]$, then flips the momentum at time $t$ and follows the trajectory 
in phase space further during the time window $[t,2t]$, one ends up with the \emph{same} initial state at time $2t$ 
after flipping the momentum again. This is at least true for all classical Hamiltonian systems in absence of any driving 
protocol ($\dot\lambda_t = 0$) and in absence of any magnetic field $B$. If a magnetic field is present, the above 
statement remains true if we also flip the magnetic field from $B$ to $-B$ during the time window $[t,2t]$. This is 
intuitively appealing if one recalls that a magnetic field is caused by \emph{moving} charges, which changes 
sign when the charges move in the opposite direction. The correct treatment of time-dependent Hamiltonians 
($\dot\lambda_t\neq0$) is revealed below. The picture above describes the essence of \emph{time-reversal symmetry}, 
which might be better called ``reversal of the direction of motion'' according to Wigner~\cite{WignerBook}. 

In quantum mechanics, one introduces a time-reversal operator $\Theta$~\cite{SakuraiBook1994}. Quite 
strangely, it turns out that the time-reversal operator has the property of being \emph{anti-unitary}, which means that 
\begin{equation}
 \langle\Theta\psi|\Theta\phi\rangle = \langle\psi|\phi\rangle^* ~~~ \text{for all} ~~~ 
 |\psi\rangle,|\phi\rangle\in\C H.
\end{equation}
Therefore, $\Theta$ is not an operator in the conventional sense (one should not intend to write it as a matrix). 
However, it follows from the above property that $\Theta$ nevertheless leaves all probabilities unchanged since 
$|\langle\Theta\psi|\Theta\phi\rangle| = |\langle\psi|\phi\rangle|$. 
It is also easy to show that anti-unitarity implies trace conjugation: 
\begin{equation}\label{eq trace conjugation}
 \mbox{tr}\{\Theta O\Theta^{-1}\} = \mbox{tr}\{O\}^* ~~~ \text{for all} ~~~ O.
\end{equation}
Furthermore, if $O$ is an observable, then $\Theta O\Theta^{-1}$ is also an observable with the same eigenvalues as 
$O$ but potentially different eigenvectors. 

As a simple example we consider the quantum mechanical treatment of particles without spin, which is in close 
analogy to the classical case. It then turns out that $\Theta$ can be identified with complex conjugation of the 
wavefunction in position representation. In equations, if $|\psi\rangle = \int dr \psi(r)|r\rangle$, where 
$|r\rangle$ are the eigenstates of the position operator $\hat r$ (here denoted with a hat to be unambiguous), then 
\begin{equation}
 \Theta|\psi\rangle = \int dr\psi^*(r)|r\rangle.
\end{equation}
Without too much effort, one confirms that 
\begin{equation}\label{eq properties Theta}
 \Theta^2 = I, ~~~ \Theta\hat r\Theta^{-1} = \hat r, ~~~ \Theta\hat p\Theta^{-1} = - \hat p,
\end{equation}
where $\hat p$ denotes the momentum operator. The properties~(\ref{eq properties Theta}) are the ones we expect by 
analogy with the classical case. Note in particular that $\Theta$ is an involution, i.e., $\Theta = \Theta^{-1}$, 
which is not always the case (more complicated systems are treated, e.g., in Ref.~\cite{HaakeBook2010}).

From what we said initially, we expect $|\psi(0)\rangle = \Theta^{-1} U_\Theta(t,0)\Theta U(t,0)|\psi(0)\rangle$ for 
any initial state $|\psi(0)\rangle$. In words: if we propagate any initial state forward in time using $U(t,0)$, then 
time-reverse it, then propagate it forward in time with respect to the time-reversed propagator $U_\Theta(t,0)$, and 
finally time-reverse it again, then we end up with the same initial state. Written as an operator identity, we have 
\begin{equation}\label{eq time reversal unitary}
 \Theta^{-1} U_\Theta(t,0)\Theta = U^\dagger(t,0).
\end{equation}
This is obviously the result one would expect mathematically, but its physical interpretation reveals an important 
symmetry. In fact, directly implementing the right hand side of this equation, i.e., $U^\dagger(t,0)$, in a lab is not 
possible as it requires to map $t\mapsto-t$. In contrast, as demonstrated below, $U_\Theta(t,0)$ corresponds to a 
legitimate `forward' evolution of a physical system. Unfortunately, however, the operator $\Theta$, being anti-unitary, 
can not be implemented in a lab in general. We now turn to the question how to define $U_\Theta(t,0)$ microscopically. 

We first consider the case of a time-independent Hamiltonian $H$ and set $U(t,0) = e^{-iHt}$ and 
$U_\Theta(t,0) = e^{-iH_\Theta t}$ with $H_\Theta$ still unknown. To infer $H_\Theta$, we use the fact that 
anti-unitarity implies anti-linearity, which means $\Theta c|\psi\rangle = c^*\Theta|\psi\rangle$ for any 
complex number $c$. From $U_\Theta(t,0) = \Theta U^\dagger(t,0)\Theta^{-1}$, we then deduce 
$H_\Theta = \Theta H\Theta^{-1}$. If $H$ denotes a Hamiltonian of interacting particles in absence of any 
magnetic field, then $H_\Theta = H$, i.e., the time-reversed motion is generated by the \emph{same} Hamiltonian. This 
follows from the fact that the momentum enters quadratically the Hamiltonian: $\Theta\hat p^2\Theta^{-1} = \hat p^2$. 
If $H = H(B)$ depends on an external magnetic field, then $H_\Theta = H(-B)$, which follows from the fact that for a 
particle with charge $q$ a term $(\hat p-qA/c)^2$ enters the Hamiltonian, where $c$ is the speed of light and $A$ the 
vector potential, which gives rise to the magnetic field. 

Finally, we consider the case with driving protocol $\lambda_s$, $s\in[0,t]$, and approximate the time evolution 
operator as 
\begin{equation}\label{eq app C Trotterization}
 U(t,0) \approx e^{-iH(\lambda_{N-1})\delta t/\hbar} \dots e^{-iH(\lambda_0)\delta t/\hbar},
\end{equation}
where we divided the time interval into steps of size $\delta t = t/N$ and implicitly keep in mind the limit 
$N\rightarrow\infty$ in which Eq.~(\ref{eq app C Trotterization}) becomes exact. We can then infer for the 
time-reversed time evolution operator 
\begin{equation}
 \begin{split}\label{eq app unitary tr}
  U_\Theta(t,0) &= \Theta U^\dagger(t,0)\Theta^{-1} \\
  &= \Theta e^{iH(\lambda_0)\delta t/\hbar} \dots e^{iH(\lambda_{N-1})\delta t/\hbar} \Theta^{-1} \\
  &= e^{-iH_\Theta(\lambda_0)\delta t/\hbar} \dots e^{-iH_\Theta(\lambda_{N-1})\delta t/\hbar},
 \end{split}
\end{equation}
where we again set $H_\Theta(\lambda_t) = \Theta H(\lambda_t)\Theta^{-1}$. Thus, the time-reversed dynamics are 
defined by changing the protocol backwards in time from $\lambda_t$ to $\lambda_0$ with respect to the 
time-reversed Hamiltonian.

\section{Basic information theory concepts}
\label{sec basics information theory}

The basic concept in quantum information theory is the von Neumann entropy 
$S_\text{vN}(\rho) = -\mbox{tr}\{\rho\ln\rho\}$. For $\rho = \sum_x \lambda _x |x\rl x|$ the von Neumann entropy reads 
\begin{equation}
 S_\text{vN}(\rho) = -\sum_x \lambda_x\ln\lambda_x \equiv S_\text{Sh}(\lambda_x),
\end{equation}
where we introduced the Shannon entropy of a classical probability distribution $\lambda_x$. Since a unitary 
transformation $U$ leaves the eigenvalues of any operator invariant, we obtain 
\begin{equation}
 S_\text{vN}(U\rho U^\dagger) = S_\text{vN}(\rho).
\end{equation}
Moreover, the von Neumann entropy is bounded by $0\le S_\text{vN}(\rho)\le\ln d$, where $d = \dim\C H$ is the 
Hilbert space dimension. 

Many other quantities in quantum (classical) information theory are closely related to the von Neumann (Shannon) 
entropy. For us important is the quantum mutual information of a bipartite state $\rho_{XY}$ 
\begin{equation}
 I_{X:Y}(\rho_{XY}) \equiv S_\text{vN}(\rho_X) + S_\text{vN}(\rho_Y) - S_\text{vN}(\rho_{XY}),
\end{equation}
which measures the amount of correlations in the state $\rho_{XY}$. It is bounded by 
\begin{equation}\label{eq MI quantum bound}
 0\le I_{X:Y}(\rho_{XY})\le 2\ln \min\{d_X,d_Y\}.
\end{equation}

By analogy, the classical mutual information for a joint probability distribution $p_{xy}$ with marginals $p_x$ and 
$p_y$ is 
\begin{equation}
 I_{X:Y}(p_{xy}) \equiv S_\text{Sh}(p_x) + S_\text{Sh}(p_y) - S_\text{Sh}(p_{xy}).
\end{equation}
It is bounded by 
\begin{equation}
 0\le I_{X:Y}(p_{xy})\le \ln \min\{d_X,d_Y\}.
\end{equation}
Note the missing factor 2 for the upper bound compared to Eq.~(\ref{eq MI quantum bound}). Furthermore, there are 
multiple ways to extend the mutual information to more than two parties with probability distribution $p_{xyz\dots}$. 
In the main text, we make twice use of the total information defined as 
\begin{equation}
 \begin{split}
  I_\text{tot}(p_{xyz\dots}) \equiv&~ S_\text{Sh}(p_x) + S_\text{Sh}(p_y) + S_\text{Sh}(p_z) + \dots \\
  &- S_\text{Sh}(p_{xyz\dots}),
 \end{split}
\end{equation}
which is always non-negative. 

A final concept used in the main text is quantum relative entropy. It is defined as 
\begin{equation}
 D(\rho\|\sigma) = \mbox{tr}\{\rho(\ln\rho-\ln\sigma)\}
\end{equation}
and measures the statistical `distance' between two states $\rho$ and $\sigma$. However, quantum relative entropy is 
not a metric since it is not symmetric: $D(\rho\|\sigma)\neq D(\sigma\|\rho)$. It satisfies $D(\rho\|\sigma) \ge 0$ 
with equality if and only if $\rho = \sigma$. Furthermore, quantum relative entropy allows to express quantum mutual 
information as 
\begin{equation}
 I_{X:Y}(\rho_{XY}) = D[\rho_{XY}\|\rho_X\otimes\rho_Y],
\end{equation}
which confirms its interpretation as a measure of correlations.

\end{document}